\documentclass[11pt]{amsart}
\usepackage{amsxtra}
\usepackage[all]{xy}
\usepackage{amssymb}
\usepackage{amsmath}
\usepackage{tikz-cd}

\theoremstyle{plain}

\newcommand{\cleqn}{\setcounter{equation}{0}}
\newcommand{\clth}{\setcounter{theorem}{0}}
\newcommand {\sectionnew}[1]{\section{#1}\cleqn\clth}

\newtheorem{theorem}{Theorem}[section]
\newtheorem{lemma}[theorem]{Lemma}
\newtheorem{definition-theorem}[theorem]{Definition-Theorem}
\newtheorem{proposition}[theorem]{Proposition}
\newtheorem{corollary}[theorem]{Corollary}
\newtheorem{definition}[theorem]{Definition}
\newtheorem{example}[theorem]{Example}
\newtheorem{remark}[theorem]{Remark}
\newtheorem{notation}[theorem]{Notation}
\newtheorem{assumption}[theorem]{Assumption}
\newtheorem{lemma-definition}[theorem]{Lemma-Definition}
\newtheorem{lemma-notation}[theorem]{Lemma-Notation}
\newtheorem{question}[theorem]{Question}
\newtheorem{remark-definition}[theorem]{Remark-Definition}
\newcommand \bth[1] { \begin{theorem}\label{t#1} }
\newcommand \ble[1] { \begin{lemma}\label{l#1} }

\newcommand \bpr[1] { \begin{proposition}\label{p#1} }
\newcommand \bco[1] { \begin{corollary}\label{c#1} }
\newcommand \bde[1] { \begin{definition}\label{d#1}\rm }
\newcommand \bex[1] { \begin{example}\label{e#1}\rm }
\newcommand \bre[1] { \begin{remark}\label{r#1}\rm }

\newcommand \bnota[1] {\begin{notation}\label{n#1}\rm }
\newcommand \bas[1] { \begin{assumption}\label{a#1}\rm }

\newcommand \bqu[1] { \begin{question}\label{q#1}\rm }

\newcommand {\ele} { \end{lemma} }

\newcommand {\epr} { \end{proposition} }
\newcommand {\eco} { \end{corollary} }
\newcommand {\ede} { \end{definition} }
\newcommand {\eex} { \end{example} }
\newcommand {\ere} { \end{remark} }
\newcommand {\enota} { \end{notation} }
\newcommand {\eas} {\end{assumption}}

\newcommand {\equ} {\end{question}}

\newcommand \lb[1]{\label{#1}}

\def \g  {\mathfrak{g}}   
\def \h  {\mathfrak{h}}

\def \n  {\mathfrak{n}}

\def \q  {\mathfrak{q}}

\newcommand{\beqa}{\begin{eqnarray*}}                     
\newcommand{\eeqa}{\end{eqnarray*}}
\def \hs {\hspace{.2in}}

\def \bfv {\underline{v}}

\def \fg  {\mathfrak{g}}
\def \fh  {\mathfrak{h}}

\def \fu  {\mathfrak{u}}

\def \fq  {\mathfrak{q}}
\def \fd  {\mathfrak{d}}

\def \bfu {{\bf u}}
\def \bfv {{\bf v}}

\def \pist {\pi_{\rm st}}

\def \sZ {{\scriptscriptstyle Z}}
\def \sC {{\scriptscriptstyle C}}
\def \sA {{\scriptscriptstyle A}}

\def \sZ {{\scriptscriptstyle Z}}

\def \lrw {\longrightarrow}
\def \Pist {\Pi_{\rm st}}

\def \wF_mn {\wF_m \times \wF_n}
\def \wF_mnC {\wF_{m, n, \, \sC}}

\def \wF {\widetilde{F}}

\begin{document}

\setlength{\baselineskip}{1.2\baselineskip}
\title[Polyubles, Poisson homogeneous spaces and multi-flag varieties]
{Polyubles, Poisson homogeneous spaces\\
 and multi-flag varieties}
\author{Shaoqiang Deng}
\address{
School of Mathematical Sciences and LPMC    \\
Nankai University \\
Tianjin 300317              \\
China}
\email{dengsq@nankai.edu.cn}

\author{Chuangchuang Kang}
\address{
College of Mathematics and Computer Science  \\
Zhejiang Normal University\\
Jinhua 321004\\}
\email{kangcc@zjnu.edu.cn}

\author{Shizhuo Yu}
\address{
School of Mathematical Sciences and LPMC    \\
Nankai University \\
Tianjin 300317              \\
China}
\email{yusz@nankai.edu.cn}

\dedicatory{\it{Dedicated to Professor Yuri I. Manin (1937-2023)}}

\begin{abstract}
A polyuble of a Manin triple can be regarded as the ``$n$-th power'' of it, which plays an important rule in the study of Poisson geometry, mathematical physics and Lie theory. In this paper, we first construct an isomorphism between the $mn$-ble and the $n$-ubles of $m$-uble by colored graph and point out it is unique. Then, we construct a class of Poisson homogeneous spaces and obtain a class of Poisson homeomorphisms between them based on the first main result. Last, we apply first two main results to multi-flag varieties as well as multi-double flag varieties and construct a class of global Poisson isomorphisms between them as well as their $T$-leaves.
\end{abstract}

\maketitle
\allowdisplaybreaks
\sectionnew{Introduction and statements of results}\lb{intro}

\subsection{Introduction}
Polyubles of Manin triples, which can be regarded as the $n$-th power  of them, were first introduced by Fock and Rosli in the study of moduli spaces of flat connections on Riemann surface \cite{poly} and further studied by J.-H. Lu and V. Mouquin in the study of Poisson homogeneous spaces and deformation quantization of Poisson algebras \cite{mixed,quanti}. Let $\mathcal{M}$ be a Manin triple, whose $n$-uple ${\mathcal M^n}$ is a Manin triple as well. First part in this paper, including $\S 2$, proves that ${\mathcal M^{mn}}$ is isomorphic to $({\mathcal M^m})^{n}$ as Manin triples and the isomorphism is unique. Furthermore, the main result in this part can be also generalized to the case of hom-Manin triple.

Manin triples naturally induce a class of quasi-triangular $r$-matrix, which provide a class of examples of Poisson manifolds and  Poisson homogeneous spaces in Lie theory. The general constructions were figured out by J.-H. Lu, V. Mouquin, and  M. Yakimov in \cite{T-leaves,orbit}. Second part in this paper, including $\S 3$ and $\S 4$, generalizes the constructions in \cite{T-leaves,orbit} such that more significant examples (see $\S$\ref{subsec:example} for more details) in Lie theory involved in our constructions. Furthermore, based on the isomorphisms in our first part, we obtain a class of Poisson diffeomomrphisms between Poisson homogeneous spaces we construct.

Let $G$ be a simply connected complex semi-simple Lie group, and let $(B, B_-)$ be a pair of opposite Borel subgroups of $G$. Multi-flag varieties
$F_n$
and multi-double varieties
$DF_{n}$
(see \S \ref{con} for precise definitions) are important research objects in the study of algebraic geometry and representations theory. In the study of Poisson geometry, J.-H. Lu and V. Mouquin construct a class of Poisson structures coming from polyubles in \cite{mixed} and describe their $T$-orbits of symplectic leaves in \cite{T-leaves}, where  $T=B\cap B_-$. Our third part, including $\S 5$, constructs a global Poisson isomorphism between $F_{2n}$ and $DF_{n}$. Furthermore, we figure out as well as the corresponding  Poisson isomorphism  between their $T$-leaves.

\subsection{Polyubles  triples and isomorphisms between them}
Recall from \cite{etingof-schiffmann,PGLecture,Manin1} that a quadratic Lie algebra is pair $(\fd,\langle,\rangle_{\fd})$, where $\fd$ is a Lie algebra and $\langle\cdot,\cdot\rangle$ is a nondegenerate symmetric invariant bilinear form. A decomposition  $\fd=\mathfrak{g}_1+\mathfrak{g}_2$  is called a Lagrangian splitting of $(\fd,\langle,\rangle_{\fd})$ if $\mathfrak{g}_1$ and $\mathfrak{g}_2$ are both Lie subalgebras of $\fd$ as well as Lagrangian subspaces of $(\fd,\langle,\rangle_{\fd})$, and the decomposition is a direct sum as vector spaces. A  Manin triples is a triple $((\mathfrak{d},\langle\cdot,\cdot\rangle),\mathfrak{g}_1,\mathfrak{g}_2)$, where $(\mathfrak{d},\langle\cdot,\cdot\rangle)$ is a quadratic Lie algebra and $\fd=\mathfrak{g}_1+\mathfrak{g}_2$ is a Lagrangian splitting.

Let $(\fd,\langle,\rangle_{\fd})$ be a quadratic Lie algebra with Lagrangian splitting $\fd=\mathfrak{g}_1+\mathfrak{g}_2$ and denote the Manin triple $((\fd,\langle,\rangle_{\fd}),\mathfrak{g}_1,\mathfrak{g}_2)$ by $\mathcal{M}$.
Let $(\fd^m,\langle,\rangle_{\fd^m})$ be the quadratic Lie algebra equipped with bilinear form
\[\langle(a_1,a_2,...,a_m),(b_1,b_2,...,b_m)\rangle_{\mathcal{M}^m}=\sum_{i=1}^{m}(-1)^{i-1}\langle a_i,b_i\rangle_\fd,\ \ a_i,b_i\in\fd.\]
Let $\fd_{\Delta}=\{(a,a):a\in\fd\}$ , which is a Lie subalgebra of $\fd\oplus\fd$. Define a Lagrangian splitting of $\fd^{m}$ as follows:
\[\fg_m=\underbrace{\fd_{\Delta}\oplus...\oplus\fd_{\Delta}}_{\frac{m}{2}},\ \fg'_{m}=\fg_2\oplus\underbrace{\fd_{\Delta}\oplus...\oplus\fd_{\Delta}}_{\frac{m}{2}-1}\oplus\fg_1,\ {\rm if\ m\ is\ even},\]

\[\fg_{m}=\underbrace{\fd_{\Delta}\oplus...\oplus\fd_{\Delta}}_{\frac{m-1}{2}}\oplus\fg_2,\ \fg'_{m}=\fg_1\oplus\underbrace{\fd_{\Delta}\oplus...\oplus\fd_{\Delta}}_{\frac{m-1}{2}},\ {\rm if\ m\ is\ odd}.\]
\textbf{Definition A} (\cite{poly})\\
The Manin triple $((\fd^m,\langle,\rangle_{\mathcal{M}^m}),\fg_m,\fg'_{m})$ is called the the $m$-uble of $\fd$, denoted by $\mathcal{M}^m$.
Based on this construction, we can continue to define the $m$-uble of the $n$-uble $\mathcal{M}^n$, denoted by $(\mathcal{M}^n)^m$, for any $m,n\in \mathbb{N}^*$. \\
\textbf{Theorem A} (Special Case of Theorem \ref{tm:mn-uble})\\
{\it There exists a unique isomorphism as a Manin triple $i^{mn}_n:\mathcal{M}^{nm}\rightarrow (\mathcal{M}^n)^m$, for any $m,n\in \mathbb{N}^*$.}\\

To prove Theorem A, we define the graph of a polyuble in the analogue of combinatorics. See \S\ref{graph} for more details. As each Manin triple induces a quasi-triangular $r$-matrix \cite{etingof-schiffmann,PGLecture,Manin1}, the graph of a polyuble also provides combinatorial interpretation for a quasi-triangular $r$-matrix induced by a polyuble. Furthermore, Theorem A can be also generalized to Theorem \ref{tm:mn-uble} in the study of hom-Lie Manin triples, which provide more information on Lie algebra and combinatorics.

\subsection{Poisson homogeneous spaces induced by polyubles}
Manin triple plays an important role in the study of Poisson geometry. Recall from \cite{PGLecture} that a Poisson manifold $(M,\pi)$ is a manifold $M$ with a  bi-vector field $\pi$ such that $\{f,g\}=\pi(df,dg)$ is a Poisson bracket. A Poisson Lie group is a Lie group $G$ with a  $\pi_G$ such that the group multiplication is a Poisson map. Recall from \cite{PGLecture} that a simply-connected Poisson Lie group $(G,\pi_G)$ is one-one corresponding to a Manin triple $((\fd,\langle,\rangle_\fd),\fg,\fg^*)$.
Let $(G,\pi_G)$ be a simply connected Poisson Lie group with Manin triple $\mathcal{M_\fg}=((\fd,\langle,\rangle_\fd),\fg,\fg^*)$. {\it A $n$-uble of a Poisson Lie group $(G,\pi_G)$} is the Poisson Lie group, denoted by  $G^{(n)}$, whose corresponding Manin triple is $\mathcal{M_\fg}^n$. Notice that $G^{(n)}$ is locally but not globally isomorphic $G^n$ in general. The version of Theorem A in Poisson Lie groups says that $(G^{(mn)},\pi_{G^{(mn)}})$ is isomorphic to $((G^{(m)})^{(n)},\pi_{(G^{(m)})^{(n)}})$ as Poisson Lie groups.

A homogeneous space $P$ of $G$ with a Poisson structure $\pi_P$ is called \emph{a Poisson homogeneous space} of Poisson Lie group $(G,\pi_G)$ if the action of  $(G,\pi_G)$ on $(P,\pi_P)$ is a Poisson map, see \cite{Drin} for reference. Not only the general theory of Poisson homogeneous spaces but also concrete examples such as flag varieties play important roles in the study of aspects such as symmetric spaces \cite{SymSpa}, symplectic groupoids \cite{SymG} and total positivity \cite{Pos}.
Let $\mathcal{M}=((\mathfrak{d},\langle,\rangle),\mathfrak{g}_1,\mathfrak{g}_2)$ be a Manin triple and $M$ a  manifold. Let $\sigma:\mathfrak{d}\rightarrow V^{1}(M)$ be the left Lie algebra action of $\mathfrak{d}$ on $M$. Let $\{\xi_{1},...,\xi_{m}\}$ be a basis of $\mathfrak{g}_2$ and $\{x_{1},...,x_{m}\}$ the dual basis of $\mathfrak{g}_1$ with respect to $\langle,\rangle$. Recall from \cite{PGLecture,Manin1} that
\begin{equation}\label{basis}
r_{\mathfrak{g}_1,\mathfrak{g}_2}=\sum_{i=1}^{m}\xi_{i}\otimes x_{i}\in \mathfrak{d}\otimes \mathfrak{d}.
\end{equation}
is  a quasi-triangular $r$-matrix of $\mathfrak{d}$. Denote $\Lambda_{\mathfrak{g}_1,\mathfrak{g}_2}\in \mathfrak{d}\wedge \mathfrak{d}$ the skew-symmetric part of $r_{\mathfrak{g}_1,\mathfrak{g}_2}$.
In \cite{T-leaves,orbit},  an equivalent condition that $\sigma(r_{\mathfrak{g}_1,\mathfrak{g}_2})$ is a Poisson structure and a sufficient condition that $\sigma(\Lambda_{\mathfrak{g}_1,\mathfrak{g}_2})$ is a Poisson structure are figured out by J.-H. Lu, V. Mouquin, M. Yakimov. Furthermore, they develop a general theory to study the regular partitions of related Poisson homogeneous spaces, including the examples of multi-flag varieties and multi-flag varieties in Lie theory.

Let $D$ be a simply-connected Lie group with Lie algebra $\mathfrak{d}$ and $I^{mn}_n:D^{mn}\rightarrow (D^m)^n$ be the unique  Lie group isomorphism induced from $I^{mn}_n$ in Theorem A. Let $Q$ be a closed Lie subgroup of $D^{mn}$ and define $\Pi^{mn}$ to be the bivector field on $D^{mn}/Q$ induced by $\mathcal{M}^{mn}$ and left translation of $D^{mn}$ . Define also $\Pi^{m,n}$ the bivector field on $(D^{m})^n/I^{mn}_n(Q)$ induced by $(M^{m})^n$ and left translation of $(D^{m})^n$. The second main theorem in this paper gives an equivalent condition that  $\sigma(\Lambda_{\mathfrak{g}_1,\mathfrak{g}_2})$ is a Poisson structure and gives a class of Poisson diffeomorphisms of Poisson homogeneous spaces based on Theorem A.\\
\textbf{Theorem B} (Special Case of Theorem \ref{Hom-Poisson} and Theorem \ref{PoissonIso})\\
(1){\it The bivector field $\sigma(\Lambda_{\mathfrak{g}_1,\mathfrak{g}_2})$ is a Poisson structure if and only if for all $ m\in M$, $[\mathfrak{q}_{m}^{\perp},\mathfrak{q}_{m}^{\perp}]\subset \mathfrak{q}_{m}$, where $\mathfrak{q}_{m}:=\{a\in \mathfrak{d}:\sigma(a)(m)=0\}$ is the stabilizer subalgebra of $m\in M$.}\\
(2){\it If the Lie algebra $\mathfrak{q}$ of $Q$ satisfies that $[\mathfrak{q}^\perp,\mathfrak{q}^\perp]\subset \mathfrak{q}$,
both $(D^{mn}/Q,\Pi^{mn})$ and  $((D^m)^n/I^{mn}_n(Q),\Pi^{m,n})$ are Poisson homogeneous spaces and they are Poisson diffeomorphic.}\\

The first part of Theorem B can be also generalized to the construction of Hom-Poisson structures on manifolds. See Theorem \ref{Hom-Poisson} for more details.
\subsection{Identification of $(DF_n,\Pi_n)$ as $(F_{2n},\pi_{2n})$}
Let now $G$ be a simply connected complex semisimple Lie groups and $(B, B_-)$ a pair of opposite Borel subgroups of $G$. Recall from \cite{T-leaves}  that multi-flag varieties
\[F_n=G \times_{B}  \cdots \times_{B} G/B\]
are the quotient spaces of twisted action of $B^n$ on $G^n$ (See (\ref{ta}) for the explicit definition) and
and multi-double varieties
\[DF_{n}=(G\times G) \times_{B\times B_-}  \cdots \times_{B\times B_-} (G\times G)/(B\times B_-)\]
are the quotient spaces of twisted action of $(B\times B_-)^n$ on $(G\times G)^n$.
In \cite{mixed}, J.-H. Lu and V. Mouquin define the Poisson structure $\pi_n$ on $F_n$ and $\Pi_n$ on $DF_n$ and figure out that they can be also obtained from polyubles and Lie algebra actions. 

Based on Theorem A and Theorem B, the following Theorem C identifies $(DF_n,\Pi_{n})$ with $(F_{2n},\pi_{2n})$ and $(DF_n,\Pi_{n})$.\\
\textbf{Theorem C} (Theorem \ref{theo:iso})\\
{\it A double multi-flag variety $(DF_n,\Pi_{n})$ is globally Poisson isomorphic to $(F_{2n},\pi_{2n})$.}

The $T$-leaves, defined in \cite{T-leaves} and referred in Definition \ref{se:Tleaf}, of $(F_{n},\pi_{n})$ and $(DF_n,\Pi_{n})$ are pointed out by J.-H. Lu and V. Mouquin in \cite{T-leaves} and play an important role in the study of  cluster algebra \cite{Shen-Weng:BS} and integrable systems \cite{Lu-Mi:Kostant}. In Poisson geometry, $T$-leaves of $(F_{2n},\pi_{2n})$ and $DF_{n}$ also provide series of Poisson groupoids and symplectic groupoids, see \cite{SymG,Mou:Guu}. Theorem C also identify $T$-leaves in $(DF_{n},\Pi_n)$ as those in $(F_{2n},\pi_{2n})$. We expect that these identifications discovered in the paper will shed new lights not only on Poisson geometry, but also on their applications to other areas of mathematics such as algebraic combinatorics, cluster algebras and quantum algebras.

\subsection{Acknowledgements}
This work was completed while the second author was supported by the Nankai University Postgraduate
Studentship and a version of Theorem A as well as
Theorem B(1) is contained in the Nankai University Ph.D. Thesis \cite{KangThesis} of the second author.  The third author would like to thank Jiang-Hua Lu and Yanpeng Li for helpful discussion.
The first author acknowledges support from the NSF China (12071228 and 12131012) and the third author acknowledges support from the NSF China (12101328).

\section{Polyubles manin triples of hom-Lie algebras and isomorphisms between them}

\subsection{Preliminaries on hom-Lie algebras}

As a generalization of Lie algebras, hom-Lie algebras  were introduced  in \cite{Hartwig}, arose from the study of q-deformations of Witt and Virasoro algebras in \cite{Hu}, which plays an important role in physics, mainly in discrete and deformed vector fields \cite{Larsson2,Larsson}.
We first recall some basic definitions and propositions of hom-Lie algebras from \cite{Benayadi,ShengBai,Sheng}.
 \begin{definition}
A {\bf hom-Lie algebra} is a triple $(\mathfrak{h},[\cdot,\cdot]_\mathfrak{h},\phi_\mathfrak{h})$ that consists of a vector space $\mathfrak{h}$, a bilinear bracket $[\cdot,\cdot]_\mathfrak{h}:\wedge^2\mathfrak{h}\rightarrow \mathfrak{h}$, and an algebra homomorphism $\phi_\mathfrak{h}:\mathfrak{h}\rightarrow \mathfrak{h}$ such that the following  hom-Jacobi identity holds:
\begin{equation}\label{eq:Hom-J}
  [\phi_\mathfrak{h}(x),[y,z]_\mathfrak{h}]_\mathfrak{h}+[\phi_\mathfrak{h}(y),[z,x]_\mathfrak{h}]_\mathfrak{h}+[\phi_\mathfrak{h}(z),[x,y]_\mathfrak{h}]_\mathfrak{h}=0, \quad \forall\, x,y,z\in \mathfrak{h}.
\end{equation}
A hom-Lie algebra is said to be {\bf involutive}, if $\phi_\mathfrak{h}^2=\mathrm{Id}$. A subspace $\mathfrak{k}\subset \mathfrak{h}$ is a {\bf hom-Lie subalgebra} of $(\mathfrak{h},[\cdot,\cdot]_\mathfrak{h},\phi_\mathfrak{h})$, if $\phi_{\mathfrak{h}}(\mathfrak{k})\subset \mathfrak{k}$ and for all $x,x'\in \mathfrak{k}$, $[x,x']_{\mathfrak{h}}\in\mathfrak{k}$.
\end{definition}

\begin{definition}
Let $(\mathfrak{h}_1,[\cdot,\cdot]_{\mathfrak{h}_1},\phi_{\mathfrak{h}_1})$ and $(\mathfrak{h}_2,[\cdot,\cdot]_{\mathfrak{h}_2},\phi_{\mathfrak{h}_2})$ be two hom-Lie algebras. A {\bf homomorphism} from  $(\mathfrak{h}_1,[\cdot,\cdot]_{\mathfrak{h}_1},\phi_{\mathfrak{h}_1})$ to $(\mathfrak{h}_2,[\cdot,\cdot]_{\mathfrak{h}_2},\phi_{\mathfrak{h}_2})$ is a linear map $f: \mathfrak{h}_1\rightarrow \mathfrak{h}_2$ such that
\begin{align}
  f\circ \phi_{\mathfrak{h}_1} & =\phi_{\mathfrak{h}_2} \circ f, \\
  f[x,y]_{\mathfrak{h}_1} & =[f(x),f(y)]_{\mathfrak{h}_2},\quad \forall\, x,y\in \mathfrak{h}_1.
\end{align}
An {\bf isomorphism} of  hom-Lie algebras is an invertible homomorphism of hom-Lie
algebras. Two hom-Lie algebras are said to be {\bf isomorphic} if there exists an
isomorphism between them.
\end{definition}

\begin{definition}
A {\bf representation} $(V,\rho,\alpha)$ of hom-Lie algebra $(\mathfrak{h},[\cdot,\cdot]_{\mathfrak{h}},\phi_{\mathfrak{h}})$ on a vector space $V$ with respect to $\alpha \in \mathfrak{gl}(V)$ is a linear map $\rho:\mathfrak{h}\rightarrow\mathfrak{gl}(V)$, such that
\begin{align}
  \rho(\phi_{\mathfrak{h}}(x))\circ \alpha & =\alpha \circ\rho(x),\\
  \rho([x,y]_{\mathfrak{h}})\circ \alpha & =\rho(\phi_{\mathfrak{h}}(x))\circ \rho(y)-\rho(\phi_{\mathfrak{h}}(y))\circ \rho(x),\quad \forall\,x,y\in \mathfrak{h}.
\end{align}
\end{definition}
Let $V^*$ be the dual vector space of $V$, $\langle\cdot,\cdot\rangle$ be the bilinear pairing between $V$ and $V^*$, $\alpha^*\in \mathfrak{gl}(V^*)$ satisfying
\begin{equation}\label{eq:phi*}
\langle\alpha(v),\xi\rangle=\langle v,\alpha^*(\xi)\rangle,\quad \forall~ v\in V,\xi\in V^*.
\end{equation}
Define a linear map $\rho^*:\mathfrak{h}\rightarrow\mathfrak{gl}(V^*)$ by
\begin{equation}\label{eq:hom-invar}
  \langle\rho^*(x)(\xi),v\rangle+\langle\xi,\rho(x)v\rangle=0,\quad \forall~ x\in\mathfrak{h}.
\end{equation}
\begin{proposition}{\rm(\cite{ShengBai})}
With the above notations, $(V^*,\rho^*,\alpha^*)$ is a representation of  hom-Lie algebra $(\mathfrak{h},[\cdot,\cdot]_{\mathfrak{h}},\phi_{\mathfrak{h}})$ if and only if the following two equations hold:
\begin{align}
  \alpha\circ\rho(\phi_{\mathfrak{h}}(x)) & = \rho(x)\circ \alpha, \label{eq:conditon-1}\\
  \alpha\circ \rho([x,y]_{\mathfrak{h}}) & = \rho(x)\circ\rho(\phi_{\mathfrak{h}}(y))-\rho(y)\circ\rho(\phi_{\mathfrak{h}}(x)),\quad\forall\,x,y\in \mathfrak{h}\label{eq:conditon-2}.
\end{align}
The triple $(V^*,\rho^*,\alpha^*)$ is called the {\bf dual representation} of hom-Lie algebra $(\mathfrak{h},[\cdot,\cdot]_{\mathfrak{h}},\phi_{\mathfrak{h}})$.
\end{proposition}

A representation $(V,\rho,\alpha)$ is called {\bf admissible} if $(V^*,\rho^*,\alpha^*)$ is also a representation, i.e. condition \eqref{eq:conditon-1} and \eqref{eq:conditon-2} in the above proposition are satisfied. It is straight forward to see that $(\mathfrak{h},\mathrm{ad},\phi_{\mathfrak{h}})$ is a representation of hom-Lie algebra $(\mathfrak{h},[\cdot,\cdot]_{\mathfrak{h}},\phi_{\mathfrak{h}})$.  Where $\mathrm{ad}:\mathfrak{h}\rightarrow \mathfrak{gl}(\mathfrak{h})$ is a linear map and for all $x,y\in\mathfrak{h}$, $\mathrm{ad}_{x}(y)=[x,y]_{\mathfrak{h}}$. We call $(\mathfrak{h},\mathrm{ad},\phi_{\mathfrak{h}})$ the {\bf adjoint representation} of hom-Lie algebra. Let $\mathfrak{h}^*$ be the dual space of $\mathfrak{h}$, $\phi_{\mathfrak{h}}^*\in \mathfrak{gl}(\mathfrak{h}^*)$. Define a linear map $\mathrm{ad}^*:\mathfrak{h}\rightarrow\mathfrak{gl}(\mathfrak{h}^*)$ by
\begin{equation}\label{eq:hom-adinvar}
  \langle \mathrm{ad}_x^*(\xi),y\rangle+\langle\xi,\mathrm{ad}_x(y)\rangle=0,\quad \forall~ x,y\in\mathfrak{h},\,\xi\in \mathfrak{h}^*.
\end{equation}
Then $(\mathfrak{h}^*,\mathrm{ad}^*,\phi_{\mathfrak{h}}^*)$ is a representation if and only if the following two equations hold:
\begin{align}
  [(\mathrm{Id}-\phi_{\mathfrak{h}}^2)(x),\phi_{\mathfrak{h}}(y)]_{\mathfrak{h}} & =0, \label{eq:ad-condition-1}\\
  [(\mathrm{Id}-\phi_{\mathfrak{h}}^2)(x),[\phi_{\mathfrak{h}}(y),z]_{\mathfrak{h}}]_{\mathfrak{h}} & = [(\mathrm{Id}-\phi_{\mathfrak{h}}^2)(y),[\phi_{\mathfrak{h}}(x),z]_{\mathfrak{h}}]_{\mathfrak{h}},\quad \forall\,x,y,z\in \mathfrak{h}\label{eq:ad-condition-2}.
\end{align}
We call $(\mathfrak{h}^*,\mathrm{ad}^*,\phi_{\mathfrak{h}}^*)$ the {\bf coadjoint representation} of hom-Lie algebra $(\mathfrak{h},[\cdot,\cdot]_{\mathfrak{h}},\phi_{\mathfrak{h}})$. And a  hom-Lie algebra is called {\bf admissible hom-Lie algebra} if  \eqref{eq:ad-condition-1} and \eqref{eq:ad-condition-2} are satisfied.
If $\phi_{\mathfrak{h}}^2=\mathrm{Id}$, then $(\mathfrak{h}^*,\mathrm{ad}^*,\phi_{\mathfrak{h}}^*)$ is a coadjoint representation of hom-Lie algebra $(\mathfrak{h},[\cdot,\cdot]_{\mathfrak{h}},\phi_{\mathfrak{h}})$, we call it {\bf involutive dual hom-Lie algebra}.

There are five different types of quadratic hom-Lie algebra in \cite{Tao2, Benayadi,ShengBai,CaiSheng,TaoBaiGuo}, but in this paper we use the definition of quadratic hom-Lie algebra in \cite{Benayadi} and \cite{ShengBai}.

\begin{definition}{\rm(\cite{Benayadi,ShengBai})}
Let $(\mathfrak{h},[\cdot,\cdot]_{\mathfrak{h}},\phi_{\mathfrak{h}})$ be a hom-Lie algebra and $\langle\cdot,\cdot\rangle_{\mathfrak{h}}:\mathfrak{h}\times \mathfrak{h}\rightarrow \mathbf{k}$ be a symmetric nondegenerate bilinear
form satisfying
\begin{align}
&\label{eq:quadra-hom-1}\langle[x,y]_{\mathfrak{h}},z\rangle_{\mathfrak{h}}=\langle x,[y,z]_{\mathfrak{h}}\rangle_{\mathfrak{h}},\\
&  \label{eq:quadra-hom-2}\langle\phi_{\mathfrak{h}}(x),y\rangle_{\mathfrak{h}}=\langle x,\phi_{\mathfrak{h}}(y)\rangle_{\mathfrak{h}},\quad\forall\,x,y,z\in\mathfrak{h}.
\end{align}
The quadruple  $(\mathfrak{h},[\cdot,\cdot]_{\mathfrak{h}},\phi_{\mathfrak{h}},\langle\cdot,\cdot\rangle_{\mathfrak{h}})$ is called a {\bf quadratic hom-Lie algebra}.
\end{definition}

\subsection{Polyubles of hom-Lie algebras }
\begin{definition}{\rm(\cite{ShengBai})}
A finite dimensional  {\bf manin triples of hom-Lie algebras} is a triple of finite dimensional hom-Lie algebras $(\mathfrak{u},\mathfrak{h},\mathfrak{h}')$. Where $(\mathfrak{u},[\cdot,\cdot]_{\mathfrak{u}},\phi_{\mathfrak{u}},\langle\cdot,\cdot\rangle_{\mathfrak{u}})$  is a quadratic hom-Lie algebra,
such that
\begin{itemize}
  \item $(\mathfrak{h},[\cdot,\cdot]_{\mathfrak{h}},\phi_{\mathfrak{h}})$ and $(\mathfrak{h}',[\cdot,\cdot]_{\mathfrak{h}'},\phi_{\mathfrak{h}'})$ are hom-Lie subalgebras of $\mathfrak{u}$ such that $\mathfrak{u}=\mathfrak{h}\oplus \mathfrak{h}'$ as vector space.
  \item $(\mathfrak{h},[\cdot,\cdot]_{\mathfrak{h}},\phi_{\mathfrak{h}})$ and $(\mathfrak{h}',[\cdot,\cdot]_{\mathfrak{h}'},\phi_{\mathfrak{h}'})$ are isotropic with respect to $\langle\cdot,\cdot\rangle_{\mathfrak{u}}$.
  \item $\phi_{\mathfrak{u}}=\phi_{\mathfrak{h}}\oplus\phi_{\mathfrak{h}'}$.
\end{itemize}
\end{definition}

\begin{definition}
 Let $(\mathfrak{u},[\cdot,\cdot]_{\mathfrak{u}},\phi_{\mathfrak{u}},\langle\cdot,\cdot\rangle_{\mathfrak{u}})$  be an $2n$-dimensional quadratic hom-Lie algebra. A hom-Lie subalgebra $(\mathfrak{h},[\cdot,\cdot]_{\mathfrak{h}},\phi_{\mathfrak{h}})$ of $(\mathfrak{u},[\cdot,\cdot]_{\mathfrak{u}},\phi_{\mathfrak{u}},\langle\cdot,\cdot\rangle_{\mathfrak{u}})$ is {\bf Lagrangian} if $\mathfrak{h}$ is maximal isotropic with respect to $\langle\cdot,\cdot\rangle_{\mathfrak{u}}$, i.e., for all $x,y \in {\mathfrak{h}}$, $\langle x,y\rangle_{\mathfrak{u}}=0$  and $\mathrm{dim}(\mathfrak{h})=n$. If both $(\mathfrak{h},[\cdot,\cdot]_{\mathfrak{h}},\phi_{\mathfrak{h}})$ and  $(\mathfrak{h}',[\cdot,\cdot]_{\mathfrak{h}'},\phi_{\mathfrak{h}'})$ are Lagrangian hom-Lie subalgebras of $(\mathfrak{u},[\cdot,\cdot]_{\mathfrak{u}},\phi_{\mathfrak{u}},\langle\cdot,\cdot\rangle_{\mathfrak{u}})$ and $\phi_{\mathfrak{u}}=\phi_{\mathfrak{h}}\oplus\phi_{\mathfrak{h}'}$, then the splitting $\mathfrak{u}=\mathfrak{h}+\mathfrak{h}'$ is called a {\bf Lagrangian splitting} of $(\mathfrak{u},[\cdot,\cdot]_{\mathfrak{u}},\phi_{\mathfrak{u}},\langle\cdot,\cdot\rangle_{\mathfrak{u}})$. And in this case, $(\mathfrak{u},\mathfrak{h},\mathfrak{h}')$ is a manin triple of hom-Lie algebras.
\end{definition}

For an integer $n\geq 1$, set $\mathfrak{u}^n=\underbrace{\mathfrak{u}\oplus\cdots\oplus\mathfrak{u}}_n$. For all $a_1,\cdots,a_n\in \mathfrak{u}$, denote its elements by $(a_1,\cdots,a_n)$. Define
a bilinear map $[\cdot,\cdot]_{\mathfrak{u}^n}:\wedge^2\mathfrak{u}^n\rightarrow \mathfrak{u}^n$  by
\begin{equation}\label{eq:direct-product-n}
 [(a_1,\cdots,a_n),(a_1', \cdots, a_n')]_{\mathfrak{u}^n}:=([a_1,a_1']_{\mathfrak{u}},\cdots, [a_n,a_n']_{\mathfrak{u}}).
\end{equation}
Define a bilinear form $\langle\cdot,\cdot\rangle_{\mathfrak{u}^n}:\mathfrak{u}^n\times\mathfrak{u}^n\rightarrow \mathbf{k}$ by
\begin{equation}\label{eq:polyuble-bilinear-n}
  \langle(a_1,a_2,\cdots,a_n),(a_1',a_2',\cdots,a_n')\rangle_{\mathfrak{u}^n}:=\sum_{j=1}^{n}(-1)^{j+1}\langle a_j,a_j'\rangle_{\mathfrak{u}}.
\end{equation}
And define an algebra homomorphism $\phi_{\mathfrak{u}^n}:\mathfrak{u}^n\rightarrow \mathfrak{u}^n$ such that
\begin{equation}\label{eq:polyuble-hom-n}
 \phi_{\mathfrak{u}^n}(a_1,a_2,\cdots,a_n):=(\phi_{\mathfrak{u}}(a_1),\phi_{\mathfrak{u}}(a_2),\cdots,\phi_{\mathfrak{u}}(a_n)), \quad \forall~ a_1,\cdots,a_n,a_1',\cdots,a_n'\in \mathfrak{u}.
\end{equation}

\begin{lemma}\label{lem:polyubles-hom}
Let $(\mathfrak{u},[\cdot,\cdot]_{\mathfrak{u}},\phi_{\mathfrak{u}},\langle\cdot,\cdot\rangle_{\mathfrak{u}})$  be an even dimensional quadratic hom-Lie algebra. With the above notations,
$(\mathfrak{u}^n,[\cdot,\cdot]_{\mathfrak{u}^n},\phi_{\mathfrak{u}^n},\langle\cdot,\cdot\rangle_{\mathfrak{u}^n})$ are  quadratic hom-Lie algebras.
\end{lemma}
\begin{proof}
For all $a_1,\cdots,a_n, a_1',\cdots,a_n'\in \mathfrak{u}$,
\begin{eqnarray*}
[(a_1,\cdots,a_n),(a_1',\cdots,a_n')]_{\mathfrak{u}^n}&\overset{\eqref{eq:direct-product-n}}{=}& (-[a_1',a_1]_{\mathfrak{u}},\cdots,-[a_n',a_n]_{\mathfrak{u}})\\
&=&-[(a_1',\cdots,a_n'),(a_1,\cdots,a_n)]_{\mathfrak{u}^n},
\end{eqnarray*}
and
\begin{eqnarray*}
\phi_{\mathfrak{u}^n}[(a_1,\cdots,a_n),(a_1',\cdots,a_n')]_{\mathfrak{u}^n}&\overset{\eqref{eq:polyuble-hom-n}}{=}& (\phi_{\mathfrak{u}}[a_1,a_1']_{\mathfrak{u}},\cdots,\phi_{\mathfrak{u}}[a_n,a_n']_{\mathfrak{u}})\\
&=&[\phi_{\mathfrak{u}^n}(a_1,\cdots,a_n),\phi_{\mathfrak{u}^n}(a_1',\cdots,a_n')]_{\mathfrak{u}^n}.
\end{eqnarray*}
By \eqref{eq:Hom-J}, for all $a_1,\cdots,a_n, a_1',\cdots,a_n',a_1'',\cdots,a_n''\in \mathfrak{u}$, we have
\begin{eqnarray*}
&&[\phi_{\mathfrak{u}^n}(a_1,\cdots,a_n),[(a_1',\cdots,a_n'),(a_1'',\cdots,a_n'')]_{\mathfrak{u}^n}]_{\mathfrak{u}^n}\\
&&+[\phi_{\mathfrak{u}^n}(a_1',\cdots,a_n'),[(a_1'',\cdots,a_n''),(a_1,\cdots,a_n)]_{\mathfrak{u}^n}]_{\mathfrak{u}^n}\\
&&+[\phi_{\mathfrak{u}^n}(a_1'',\cdots,a_n''),[(a_1,\cdots,a_n),(a_1',\cdots,a_n')]_{\mathfrak{u}^n}]_{\mathfrak{u}^n}\\
&=& [(\phi_{\mathfrak{u}}(a_1),\cdots,\phi_{\mathfrak{u}}(a_n)),([a_1',a_1'']_{\mathfrak{u}},\cdots,[a_n',a_n'']_{\mathfrak{u}})]_{\mathfrak{u}^n}\\
&&+[(\phi_{\mathfrak{u}}(a_1'),\cdots,\phi_{\mathfrak{u}}(a_n')),([a_1'',a_1]_{\mathfrak{u}},\cdots,[a_n'',a_n]_{\mathfrak{u}})]_{\mathfrak{u}^n}\\
&&+[(\phi_{\mathfrak{u}}(a_1''),\cdots,\phi_{\mathfrak{u}}(a_n'')),([a_1,a_1']_{\mathfrak{u}},\cdots,[a_n,a_n']_{\mathfrak{u}})]_{\mathfrak{u}^n}\\
&=&(0,\cdots,0).
\end{eqnarray*}
Then $(\mathfrak{u}^n,[\cdot,\cdot]_{\mathfrak{u}^n},\phi_{\mathfrak{u}^n})$ are hom-Lie algebras.

Moreover, we have
\begin{eqnarray*}
\langle\phi_{\mathfrak{u}^n}(a_1,\cdots,a_n),(a_1',\cdots,a_n')\rangle_{\mathfrak{u}^n}
&\overset{\eqref{eq:polyuble-hom-n}}{=}&\langle(\phi_{\mathfrak{u}}(a_1),\cdots,\phi_\mathfrak{u}(a_n)),(a_1',\cdots,a_n')\rangle_{\mathfrak{u}^n}\\
&\overset{\eqref{eq:polyuble-bilinear-n}}{=}&\sum_{j=1}^{n}(-1)^{j+1}\langle \phi_{\mathfrak{u}}(a_j),a_j'\rangle_{\mathfrak{u}}\\
&\overset{\eqref{eq:quadra-hom-2}}{=}&\sum_{j=1}^{n}(-1)^{j+1}\langle a_j,\phi_{\mathfrak{u}}(a_j')\rangle_{\mathfrak{u}}\\
&=&\langle(a_1,\cdots,a_n),\phi_{\mathfrak{u}^n}(a_1',\cdots,a_n')\rangle_{\mathfrak{u}^n},
\end{eqnarray*}
and
\begin{eqnarray*}
    &&\langle[(a_1,\cdots,a_n),(a_1',\cdots,a_n')]_{\mathfrak{u}^n},(a_1'',\cdots,a_n'')\rangle_{\mathfrak{u}^n}\\
   &\overset{\eqref{eq:direct-product-n}}{=}&\langle([a_1,a_1']_{\mathfrak{u}},\cdots,[a_n,a_n']_{\mathfrak{u}}),(a_1'',\cdots,a_n'')\rangle_{\mathfrak{u}^n}\\
   &\overset{\eqref{eq:polyuble-bilinear-n}}{=}&  \sum_{j=1}^{n}(-1)^{j+1}\langle [a_j,a_j']_{\mathfrak{u}},a_j''\rangle_{\mathfrak{u}} \\
   &=& \langle(a_1,\cdots,a_n),[(a_1',\cdots,a_n'),(a_1'',\cdots,a_n'')]_{\mathfrak{u}^n}\rangle_{\mathfrak{u}^n}.
\end{eqnarray*}
Therefore, $(\mathfrak{u}^n,[\cdot,\cdot]_{\mathfrak{u}^n},\phi_{\mathfrak{u}^n},\langle\cdot,\cdot\rangle_{\mathfrak{u}^n})$ are quadratic hom-Lie algebras.
\end{proof}

\begin{theorem}\label{tm:hom-polyubles}
Let $(\mathfrak{u},[\cdot,\cdot]_{\mathfrak{u}},\phi_{\mathfrak{u}},\langle\cdot,\cdot\rangle_{\mathfrak{u}})$  be an even dimensional quadratic hom-Lie algebra with the Lagrangian splitting $\mathfrak{u}=\mathfrak{h}+\mathfrak{h}'$ and let $\mathfrak{u}_{\Delta}$ be the diagonal hom-Lie subalgebra of $\mathfrak{u}\oplus \mathfrak{u}$, i.e.
\begin{align}
&\mathfrak{u}_{\Delta}=\{(a,a)~|~a\in\mathfrak{u}\}\label{eq:diag}.
\end{align}
For an integer $n\geq 1$, define a Lagrangian splitting of $\mathfrak{u}^n$ as follows:
\begin{align}
  \mathfrak{h}_{n} & := \underbrace{\mathfrak{u}_{\Delta}\oplus\cdots\oplus\mathfrak{u}_{\Delta}}_{\frac{n-1}{2}}\oplus\mathfrak{h},\quad \mathfrak{h}_{n}':=\mathfrak{h}'\oplus \underbrace{\mathfrak{u}_{\Delta}\oplus\cdots\oplus\mathfrak{u}_{\Delta}}_{\frac{n-1}{2}},\quad\mathrm{if~n~is~odd},
  \label{eq:split-odd}\\
  \mathfrak{h}_{n} & := \underbrace{\mathfrak{u}_{\Delta}\oplus\cdots\oplus\mathfrak{u}_{\Delta}}_{\frac{n}{2}},\quad \mathfrak{h}_{n}':=\mathfrak{h}'\oplus\underbrace{\mathfrak{u}_{\Delta}\oplus\cdots\oplus\mathfrak{u}_{\Delta}}_{\frac{n}{2}-1}\oplus\mathfrak{h},\quad \mathrm{if~n~is~even}\label{eq:split-even}.
\end{align}
Then $(\mathfrak{u}^n,\mathfrak{h}_{n},\mathfrak{h}_{n}')$ are manin triples of hom-Lie algebras.
We call $$((\mathfrak{u}^{n},[\cdot,\cdot]_{\mathfrak{u}^{n}},\phi_{\mathfrak{u}^{n}},\langle\cdot,\cdot\rangle_{\mathfrak{u}^n}),\mathfrak{h}_{n},\mathfrak{h}_{n}')$$ the {\bf polyubles (or $n$-uble) of hom-Lie algebras} and which is denoted by  $\mathcal{M}_{\phi}^n$.
\end{theorem}
\begin{proof}
When $n$ is odd. For all $a_1,\cdots,a_{\frac{n-1}{2}}, b_1,\cdots,b_{\frac{n-1}{2}} \in \mathfrak{u},x\in \mathfrak{h}, x'\in \mathfrak{h}'$, define two  algebra homomorphisms $\phi_{\mathfrak{h}_{n}}:\mathfrak{h}_{n}\rightarrow \mathfrak{h}_{n}$ and $\phi_{\mathfrak{h}_{n}'}:\mathfrak{h}_{n}'\rightarrow \mathfrak{h}_{n}'$ by
\begin{align*}
\phi_{\mathfrak{h}_{n}}((a_1,a_1),\cdots,(a_{\frac{n-1}{2}},a_{\frac{n-1}{2}}),x)
=((\phi_{\mathfrak{u}}(a_1),\phi_{\mathfrak{u}}(a_1)),
\cdots,(\phi_{\mathfrak{u}}(a_{\frac{n-1}{2}}),\phi_{\mathfrak{u}}(a_{\frac{n-1}{2}})),\phi_{\mathfrak{h}}(x)),
\end{align*}
and
\begin{align*}
\phi_{\mathfrak{h}_{n}'}(x',(b_1,b_1),\cdots,(b_{\frac{n-1}{2}},b_{\frac{n-1}{2}}))
=(\phi_{\mathfrak{h}}'(x'),(\phi_{\mathfrak{u}}(b_1),\phi_{\mathfrak{u}}(b_1)),
\cdots,(\phi_{\mathfrak{u}}(b_{\frac{n-1}{2}}),\phi_{\mathfrak{u}}(b_{\frac{n-1}{2}}))).
\end{align*}
Then we have $\mathrm{dim}(\mathfrak{h}_{n})=\mathrm{dim}(\mathfrak{h}_{n}')$,
$\mathfrak{u}^{n}= \mathfrak{h}_{n}\oplus \mathfrak{h}_{n}'$ and $\phi_{\mathfrak{u}^n}=\phi_{\mathfrak{h}_{n}}\oplus \phi_{\mathfrak{h}_{n}'}$.

For all $c_1,\cdots,c_{\frac{n-1}{2}}$, $d_1,\cdots,d_{\frac{n-1}{2}} \in \mathfrak{u}$, $y\in \mathfrak{h}, y'\in \mathfrak{h}'$, define two bilinear maps $[\cdot,\cdot]_{\mathfrak{h}_n}:\wedge^2\mathfrak{h}_n\rightarrow \mathfrak{h}_n$ and
$[\cdot,\cdot]_{\mathfrak{h}_n'}:\wedge^2\mathfrak{h}_n'\rightarrow \mathfrak{h}_n'$  by
\begin{align*}
&[((a_1,a_1),\cdots,(a_{\frac{n-1}{2}},a_{\frac{n-1}{2}}),x),((c_1,c_1),\cdots,(c_{\frac{n-1}{2}},c_{\frac{n-1}{2}}),y)]_{\mathfrak{h}_n}\\
=&(([a_1,c_1]_{\mathfrak{u}},[a_1,c_1]_{\mathfrak{u}}),\cdots,([a_{\frac{n-1}{2}},c_{\frac{n-1}{2}}]_{\mathfrak{u}},
[a_{\frac{n-1}{2}},c_{\frac{n-1}{2}}]_{\mathfrak{u}}),[x,y]_{\mathfrak{h}}),
\end{align*}
and
\begin{align*}
&[(x',(b_1,b_1),\cdots,(b_{\frac{n-1}{2}},b_{\frac{n-1}{2}})),(y',(d_1,d_1),\cdots,(d_{\frac{n-1}{2}},d_{\frac{n-1}{2}}))]_{\mathfrak{h}_n}\\
=&([x',y']_{\mathfrak{h}}([b_1,d_1]_{\mathfrak{u}},[b_1,d_1]_{\mathfrak{u}}),\cdots,([b_{\frac{n-1}{2}},d_{\frac{n-1}{2}}]_{\mathfrak{u}},
[b_{\frac{n-1}{2}},d_{\frac{n-1}{2}}]_{\mathfrak{u}})).
\end{align*}
Then $\phi_{\mathfrak{u}^n}(\mathfrak{h}_n)\subset \mathfrak{h}_n$ and $\phi_{\mathfrak{u}^n}(\mathfrak{h}_n')\subset \mathfrak{h}_n'$. Furthermore, we have $$[((a_1,a_1),\cdots,(a_{\frac{n-1}{2}},a_{\frac{n-1}{2}}),x),
((c_1,c_1),\cdots,(c_{\frac{n-1}{2}},c_{\frac{n-1}{2}}),y)]_{\mathfrak{u}^n}\in \mathfrak{h}_n,$$
and
$$[(x',(b_1,b_1),\cdots,(b_{\frac{n-1}{2}},b_{\frac{n-1}{2}})),
(y',(c_1,c_1),\cdots,(c_{\frac{n-1}{2}},c_{\frac{n-1}{2}}))]_{\mathfrak{u}^n}\in \mathfrak{h}_n'.$$
Then $(\mathfrak{h}_n,[\cdot,\cdot]_{\mathfrak{h}_n},\phi_{\mathfrak{h}_n})$ and $(\mathfrak{h}_n',[\cdot,\cdot]_{\mathfrak{h}_n'},\phi_{\mathfrak{h}_n'})$ are hom-Lie subalgebras of $(\mathfrak{u}^{n},[\cdot,\cdot]_{\mathfrak{u}^{n}},\phi_{\mathfrak{u}^{n}})$.

By Lemma \ref{lem:polyubles-hom}, $(\mathfrak{u}^{n},[\cdot,\cdot]_{\mathfrak{u}^{n}},\phi_{\mathfrak{u}^{n}},\langle\cdot,\cdot\rangle_{\mathfrak{u}^n})$ are quadratic hom-Lie algebras. Since  $\mathfrak{u}=\mathfrak{h}+\mathfrak{h}'$ is a Lagrangian splitting, we have
$$
\langle((a_1,a_1),\cdots,(a_{\frac{n-1}{2}},a_{\frac{n-1}{2}}),x),
((c_1,c_1),\cdots,(c_{\frac{n-1}{2}},c_{\frac{n-1}{2}}),y)\rangle_{\mathfrak{u}^n}=0,
$$
and
$$
\langle (x',(b_1,b_1),\cdots,(b_{\frac{n-1}{2}},b_{\frac{n-1}{2}})),
(y',(c_1,c_1),\cdots,(c_{\frac{n-1}{2}},c_{\frac{n-1}{2}}))\rangle_{\mathfrak{u}^n}=0.
$$
Which implies that $(\mathfrak{h}_n,[\cdot,\cdot]_{\mathfrak{h}_n},\phi_{\mathfrak{h}_n})$ and $(\mathfrak{h}_n',[\cdot,\cdot]_{\mathfrak{h}_n'},\phi_{\mathfrak{h}_n'})$ are isotropic with respect to $\langle\cdot,\cdot\rangle_{\mathfrak{u}^n}$.

Therefore,
$(\mathfrak{u}^n,\mathfrak{h}_{n},\mathfrak{h}_{n}')$ are manin triples of hom-Lie algebras.
Similarly, when $n$ is even, $(\mathfrak{u}^n,\mathfrak{h}_{n},\mathfrak{h}_{n}')$ are manin triples of hom-Lie algebras. Hence the conclusion holds.
\end{proof}
\subsection{Graph of polyubles and proof of Theorem A}\label{graph}
In graph theory, an undirected graph $G$ is a pair $(V,E)$ where $V$ is the set of vertices and  $E$ is the set of edges connecting adjacency vertices. An isomorphism of graphs $G=(V,E)$ and $H=(V',E')$ is an edge-preserving bijection $f:V\rightarrow V'$ such that any two vertices $u$ and $v$ of $G$ are adjacent in $G$ if and only if $f(u)$ and $f(v)$ are adjacent in $H$. If an isomorphism exists between two graphs, then the graphs are called isomorphic and denoted as $G\simeq H$.

Represent $\mathfrak{u}$, $\mathfrak{h}$  and $\mathfrak{h}'$ with vertices $\circ$, $\triangleleft$  and $\triangleright$, respectively. Then every $\circ$ can be represented by $\triangleleft\oplus \triangleright$. For the bilinear form $\langle\cdot,\cdot\rangle_{\mathfrak{u}^n}$ defined as \eqref{eq:polyuble-bilinear-n}, if the sign in front of $\langle a_j,a_j'\rangle_{\mathfrak{u}}$ is negative,  then we use the symbol $\bullet$ or $\blacktriangleleft$ or $\blacktriangleright$, if the sign  is  positive, then we use the symbol $\circ$ or $\triangleleft$ or $\triangleright$. Represent $\mathfrak{u}_{\Delta}$ with the edge $\xymatrix@C=0.5cm{
  \circ \ar@{-}[r] & \circ}$.  The $i$-th $\mathfrak{u}$ is denoted by the symbol  $\circ_i$.
Hence $\mathcal{M}_{\phi}^{n}$ can  be expressed as a chain of $\circ_i,~i=1,\cdots,n$, i.e.
if $n$ is odd, then we can represent the $n$-uble of hom-Lie algebras as:
 $$
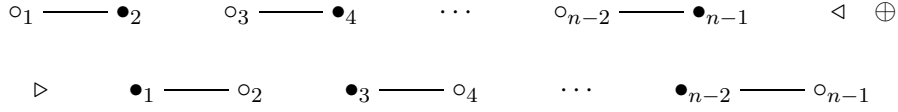
\xymatrix{
  \circ_1 \ar@{-}[r] & \bullet_2  & \circ_3 \ar@{-}[r] &\bullet_4 &\cdots & \circ_{n-2} \ar@{-}[r] & \bullet_{n-1}  & \triangleleft \quad\oplus}
 $$
 $$
 \xymatrix{
  \triangleright & \bullet_1 \ar@{-}[r] & \circ_2  & \bullet_3 \ar@{-}[r] & \circ_4 &\cdots & \bullet_{n-2} \ar@{-}[r] & \circ_{n-1}  }
 $$
If $n$ is even, then we can represent the $n$-uble of hom-Lie algebras as:
 $$
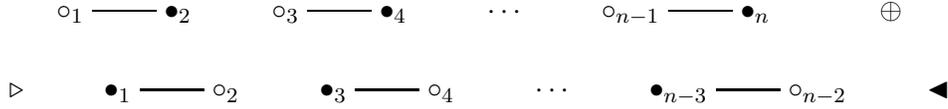
\xymatrix{
  \circ_1 \ar@{-}[r] & \bullet_2  & \circ_3 \ar@{-}[r] &\bullet_4 &\cdots & \circ_{n-1} \ar@{-}[r] & \bullet_{n}  & \quad\oplus}
 $$
 $$
 \xymatrix{
  \triangleright & \bullet_1 \ar@{-}[r] & \circ_2  & \bullet_3 \ar@{-}[r] & \circ_4 &\cdots & \bullet_{n-3} \ar@{-}[r] & \circ_{n-2}  & \blacktriangleleft  }
 $$

An isomorphism of two Manin triples of hom-Lie algebras $(\mathfrak{u}_1,\mathfrak{h}_1,\mathfrak{h}_{1}')$ and $(\mathfrak{u}_2,\mathfrak{h}_2,\mathfrak{h}_{2}')$ is a one-to-one linear map $f:\mathfrak{u}_1\rightarrow \mathfrak{u}_2$ such that
\begin{align}
f[x,y]_{\mathfrak{u}_1}&=[f(x),f(y)]_{\mathfrak{u}_2},\label{eq:Lie-I}\\
\langle x,y\rangle_{\mathfrak{u}_1}&=\langle f(x),f(y)\rangle_{\mathfrak{u}_2},\label{eq:B-I}\\
f\circ \phi_{\mathfrak{u}_1}&=\phi_{\mathfrak{u}_2}\circ f,\quad\forall~x,y\in \mathfrak{u}_1.\label{eq:phi-I}
\end{align}
By the graph of $n$-uble of hom-Lie algebras, any two adjacent vertices are different colors. If two $n$-uble of hom-Lie algebras are  isomorphic, then  there exists a unique bijection between two corresponding graphs.
Using the isomorphism of graphs, we obtain the following conclusion.

\begin{theorem}\label{tm:mn-uble}
For any $m,n\in \mathbb{N}^*$, there exists a unique isomorphism as a Manin triple of hom-Lie algebras $i^{mn}_n:\mathcal{M}_{\phi}^{mn}\rightarrow (\mathcal{M}_{\phi}^m)^n$,
i.e. $\mathcal{M}_{\phi}^{mn}$ are isomorphic to $(\mathcal{M}_{\phi}^m)^n$ and  $(\mathcal{M}_{\phi}^n)^m$.
\end{theorem}

\begin{proof}
For any $m,n\in \mathbb{N}^*$,
$(\mathcal{M}_{\phi}^m)^n$ are represented by a column matrix $(\overbrace{\mathcal{M}_{\phi}^m,\cdots,\mathcal{M}_{\phi}^m}^n)^{T}$, then $(\mathcal{M}_{\phi}^m)^n$ can be seen as a matrix of $n$ rows and $m$ columns of $\mathfrak{u}$.
For all $a_{1},\cdots,a_{mn} \in \mathfrak{u}$, when $m$ is odd, define a linear map $i^{mn}_n: \mathcal{M}_{\phi}^{mn} \rightarrow (\mathcal{M}_{\phi}^m)^n$ by:
\begin{eqnarray*}
&&i^{mn}_n(a_1,\cdots,a_n,a_{n+1},\cdots,a_{2n},\cdots,a_{(m-1)n+1}\cdots,a_{mn})\\
&=&\left(
   \begin{array}{ccccccc}
     a_1 & a_{2n} & a_{2n+1} & a_{4n} & \cdots & a_{(m-1)n+1}\\
     a_2 & a_{2n-1} & a_{2n+2} & a_{4n-1} & \cdots & a_{(m-1)n+2} \\
     \vdots & \vdots & \vdots & \vdots   &\ddots &\vdots\\
    a_{n} & a_{n+1} & a_{3n} & a_{3n+1} & \cdots & a_{mn}\\
   \end{array}
 \right).
\end{eqnarray*}
When $m$ is even, define a linear map $i^{mn}_n: \mathcal{M}_{\phi}^{mn} \rightarrow (\mathcal{M}_{\phi}^m)^n$ by:
\begin{eqnarray*}
&&i^{mn}_n(a_1,\cdots,a_n,a_{n+1},\cdots,a_{2n},\cdots,a_{(m-1)n+1}\cdots,a_{mn})\\
&=&\left(
   \begin{array}{ccccccc}
     a_1 & a_{2n} & a_{2n+1} & a_{4n} & \cdots & a_{mn}\\
     a_2 & a_{2n-1} & a_{2n+2} & a_{4n-1} & \cdots & a_{mn-1} \\
     \vdots & \vdots & \vdots & \vdots   &\ddots &\vdots\\
    a_{n} & a_{n+1} & a_{3n} & a_{3n+1} & \cdots & a_{(m-1)n+1}\\
   \end{array}
 \right).
\end{eqnarray*}

According to definition of $i^{mn}_n$, the last column of $(\mathcal{M}_{\phi}^m)^n$ are determined by the parity of $m$. And the map $i^{mn}_n$ change the order of the chains in the shape of a ``snake'', i,e.
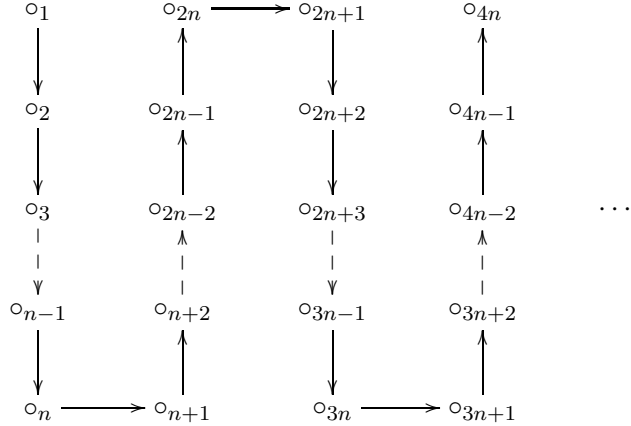
$$
\xymatrix{
  \circ_1 \ar[d]_{}  &  \ar[r]^{}\circ_{2n} & \circ_{2n+1} \ar[d]^{}& \circ_{4n}  &~\\
  \circ_2 \ar[d]_{}  & \circ_{2n-1} \ar[u]_{}  & \circ_{2n+2} \ar[d]^{}& \circ_{4n-1} \ar[u]^{} &~\\
  \circ_3 \ar@{-->}[d]   & \circ_{2n-2} \ar[u]_{}  & \circ_{2n+3} \ar@{-->}[d]^{}& \circ_{4n-2} \ar[u]^{}&\cdots \\
  \circ_{n-1} \ar[d]_{}  & \circ_{n+2} \ar@{-->}[u]_{}  & \circ_{3n-1} \ar[d]^{}& \circ_{3n+2} \ar@{-->}[u]^{} &~\\
  \circ_n \ar[r]_{} & \circ_{n+1} \ar[u]_{}  & \circ_{3n}\ar[r]_{}  & \circ_{3n+1}\ar[u]^{}  &~}
$$
Hence $i^{mn}_n$ is a edge-preserving bijection.

Define $\mathfrak{u}_{\Delta}^m$ by
\begin{equation}\label{eq:diag-m}
 \mathfrak{u}_{\Delta}^m:=\{(a_{1},\cdots,a_{m},a_{m+1},\cdots,a_{2m})~|~a_{i}=a_{m+i}\in \mathfrak{u},~\forall~1\leq i\leq m\}.
\end{equation}
Then we can express $\mathfrak{u}_{\Delta}^m$ with the following graph:
$$
 
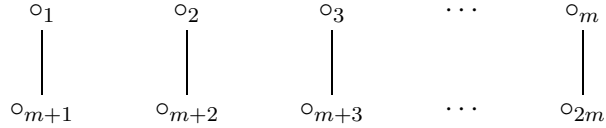
\xymatrix{
  \circ_1 \ar@{-}[d]  &  \circ_{2} & \circ_{3} \ar@{-}[d]   & \cdots&\circ_{m}  &\\
  \circ_{m+1}           & \circ_{m+2} \ar@{-}[u]  & \circ_{m+3} &   \cdots&\circ_{2m} \ar@{-}[u]&
  }
$$

{\bf Case 1:} When $m,n$ are both even. On the one hand,
$$\mathcal{M}_{\phi}^{mn}=((\mathfrak{u}^{mn},[\cdot,\cdot]_{\mathfrak{u}^{mn}},\phi_{\mathfrak{u}^{mn}},\langle\cdot,\cdot\rangle_{\mathfrak{u}^{mn}}),\mathfrak{h}_{mn},\mathfrak{h}_{mn}') $$
are given by
\begin{equation}\label{eq:bilinear-mn}
  \langle(a_1,\cdots,a_{mn}),(a_1',\cdots,a_{mn}')\rangle_{\mathfrak{u}^{mn}}=\sum_{j=1}^{mn}(-1)^{j+1}\langle a_j,a_j'\rangle_{\mathfrak{u}},
\end{equation}
\begin{equation}\label{eq:phi-mn}
 \phi_{\mathfrak{u}^{mn}}(a_1,\cdots,a_{mn})=(\phi_{\mathfrak{u}}(a_1),\cdots,\phi_{\mathfrak{u}}(a_{mn})), \quad \forall~ a_{j},a_{j}'\in \mathfrak{u},
\end{equation}
\begin{equation}\label{eq:spacesplit-mn}
  \mathfrak{h}_{mn} = \underbrace{\mathfrak{u}_{\Delta}\oplus\cdots\oplus\mathfrak{u}_{\Delta}}_{\frac{mn}{2}},\quad \mathfrak{h}_{mn}'=\mathfrak{h}'\oplus\underbrace{\mathfrak{u}_{\Delta}\oplus\cdots\oplus\mathfrak{u}_{\Delta}}_{\frac{mn}{2}-1}\oplus\mathfrak{h}.
\end{equation}
On the other hand, $$(\mathcal{M}_{\phi}^{m})^n=((\mathfrak{u}^{mn},[\cdot,\cdot]_{\mathfrak{u}^{mn}},\phi_{(\mathfrak{u}^{m})^n},\langle\cdot,\cdot\rangle_{(\mathfrak{u}^{m})^n}),\tilde{\mathfrak{h}}_{mn},\tilde{\mathfrak{h}}_{mn}') $$ are given by
\begin{equation}\label{eq:mn-bilinear}
\langle(a_1,\cdots,a_{mn}),(a_1',\cdots,a_{mn}')\rangle_{(\mathfrak{u}^{m})^n}=\sum_{k=0}^{n}\sum_{j=km+1}^{(k+1)m}(-1)^{k+j-1}\langle a_j,a_j'\rangle_{\mathfrak{u}},
\end{equation}
\begin{equation}\label{eq:mn-hom}
\phi_{(\mathfrak{u}^{m})^n}(a_1,\cdots,a_{mn})=(\underbrace{(\phi_{\mathfrak{u}}(a_1),\cdots,\phi_{\mathfrak{u}}(a_m)),\cdots,(\phi_{\mathfrak{u}}(a_1),\cdots,\phi_{\mathfrak{u}}(a_m))}_{n}), \quad \forall~ a_{j},a_{j}'\in \mathfrak{u},
\end{equation}
\begin{equation}\label{mn-split-even-even1}
\tilde{\mathfrak{h}}_{mn}=\underbrace{\mathfrak{u}_{\Delta}^m \oplus \cdots\oplus\mathfrak{u}_{\Delta}^m}_{\frac{n}{2}},
\end{equation}
\begin{equation}\label{mn-split-even-even2}
\tilde{\mathfrak{h}}_{mn}' = \mathfrak{h}'\oplus\underbrace{\mathfrak{u}_{\Delta}\oplus\cdots\oplus\mathfrak{u}_{\Delta}}_{\frac{m}{2}-1}\oplus\mathfrak{h}\oplus\underbrace{\mathfrak{u}_{\Delta}^m \oplus \cdots\oplus\mathfrak{u}_{\Delta}^m}_{\frac{n}{2}-1}\oplus
\underbrace{\mathfrak{u}_{\Delta}\oplus\cdots\oplus\mathfrak{u}_{\Delta}}_{\frac{m}{2}}.
\end{equation}

By the graph of polyubles of hom-Lie algebras, we have
$$
i^{mn}_n(\xymatrix{
  \circ_1 \ar@{-}[r] & \bullet_2  & \circ_3 \ar@{-}[r] & \bullet_4 &\cdots & \circ_{mn-1} \ar@{-}[r] & \bullet_{mn}  \quad\oplus}
 $$
 $$
 \xymatrix{
  \triangleleft &\bullet_1 \ar@{-}[r] & \circ_2  & \bullet_3 \ar@{-}[r] & \circ_4 &\cdots & \bullet_{mn-3} \ar@{-}[r] & \circ_{mn-2}  &\blacktriangleright})=
 $$
 $$
 \xymatrix{
  \circ_1 \ar@{-}[d]  &  \bullet_{2n} & \circ_{2n+1} \ar@{-}[d]   & \cdots&\bullet_{mn}  &\\
  \bullet_2           & \circ_{2n-1} \ar@{-}[u]  & \bullet_{2n+2} &   \cdots&\circ_{mn-1} \ar@{-}[u]& \\
  \vdots              & \vdots  & \vdots& \ddots & \vdots& \oplus\\
  \circ_{n-1} \ar@{-}[d]  & \bullet_{n+2}  & \circ_{3n-1} \ar@{-}[d]& \cdots&\bullet_{(m-1)n+2}\ar@{-}[d]  &\\
  \bullet_n           & \circ_{n+1} \ar@{-}[u]  & \bullet_{3n}    & \cdots&\circ_{(m-1)n+1} & }
 $$
 \\
 \\
$$
 \xymatrix{
   \triangleleft   &  \bullet_{2n-1}\ar@{-}[r] & \circ_{2n}    & \cdots&\blacktriangleright  &\\
  \bullet_1 \ar@{-}[d]  &  \circ_{2n-2} & \bullet_{2n+1} \ar@{-}[d]   & \cdots&\circ_{mn-2}  &\\
  \circ_2           & \bullet_{2n-3} \ar@{-}[u]  & \circ_{2n+2} &   \cdots&\bullet_{mn-3} \ar@{-}[u]& \\
  \vdots              & \vdots  & \vdots& \ddots & \vdots&\\
  \bullet_{n-3} \ar@{-}[d]  & \circ_{n+2}  & \bullet_{3n-3} \ar@{-}[d]& \cdots&\circ_{(m-1)n+2}\ar@{-}[d]  &\\
  \circ_{n-2}           & \bullet_{n+1} \ar@{-}[u]  & \circ_{3n-2}    & \cdots&\bullet_{(m-1)n+1} & \\
   \bullet_{n-1} \ar@{-}[r]          & \circ_{n}  & \bullet_{3n-1}    & \cdots& \circ_{(m-1)n}&}
$$
Therefore, $i^{mn}_n(\mathfrak{h}_{mn}\oplus \mathfrak{h}_{mn}')=\tilde{\mathfrak{h}}_{mn}\oplus \tilde{\mathfrak{h}}_{mn}'$.

{\bf Case 2:} When $m,n$ are both odd. On the one hand,
$$\mathcal{M}_{\phi}^{mn}=((\mathfrak{u}^{mn},[\cdot,\cdot]_{\mathfrak{u}^{mn}},\phi_{\mathfrak{u}^{mn}},\langle\cdot,\cdot\rangle_{\mathfrak{u}^{mn}}),\mathfrak{h}_{mn},\mathfrak{h}_{mn}') $$
are given by
\begin{equation}
 \mathfrak{h}_{mn} = \mathfrak{h}\oplus\underbrace{\mathfrak{u}_{\Delta}\oplus\cdots\oplus\mathfrak{u}_{\Delta}}_{\frac{mn-1}{2}},\quad \mathfrak{h}_{mn}'=\underbrace{\mathfrak{u}_{\Delta}\oplus\cdots\oplus\mathfrak{u}_{\Delta}}_{\frac{mn-1}{2}}\oplus \mathfrak{h}',
\end{equation}
 $\phi_{\mathfrak{u}^{mn}}$ and $\langle\cdot,\cdot\rangle_{\mathfrak{u}^{mn}}$ keep the same with the first case.
On the other hand, $$(\mathcal{M}_{\phi}^{m})^n=((\mathfrak{u}^{mn},[\cdot,\cdot]_{\mathfrak{u}^{mn}},\phi_{(\mathfrak{u}^{m})^n},\langle\cdot,\cdot\rangle_{(\mathfrak{u}^{m})^n}),\tilde{\mathfrak{h}}_{mn},\tilde{\mathfrak{h}}_{mn}') $$
are given by
\begin{align}
\tilde{\mathfrak{h}}_{mn}&=\mathfrak{h}\oplus\underbrace{\mathfrak{u}_{\Delta}^m \oplus \cdots\oplus\mathfrak{u}_{\Delta}^m}_{\frac{n-1}{2}}\oplus\underbrace{\mathfrak{u}_{\Delta}\oplus\cdots\oplus\mathfrak{u}_{\Delta}}_{\frac{m-1}{2}},\label{mn-split-odd-odd1}\\
 \tilde{\mathfrak{h}}_{mn}' &= \underbrace{\mathfrak{u}_{\Delta}\oplus\cdots\oplus\mathfrak{u}_{\Delta}}_{\frac{m-1}{2}}\oplus\underbrace{\mathfrak{u}_{\Delta}^m \oplus \cdots\oplus\mathfrak{u}_{\Delta}^m}_{\frac{n-1}{2}}\oplus\mathfrak{h}',\label{mn-split-odd-odd2}
\end{align}
$\phi_{(\mathfrak{u}^{m})^n}$ and $\langle\cdot,\cdot\rangle_{(\mathfrak{u}^{m})^n}$ keep the same with the first case. Then we have

$$
i^{mn}_n(\xymatrix{
  \triangleleft& \bullet_1 \ar@{-}[r] & \circ_2  & \bullet_3 \ar@{-}[r] & \circ_4 &\cdots & \bullet_{mn-2} \ar@{-}[r] & \circ_{mn-1}  \quad\oplus}
 $$
 $$
 \xymatrix{
\circ_1 \ar@{-}[r] & \bullet_2  & \circ_3 \ar@{-}[r] & \bullet_4 &\cdots & \circ_{mn-2} \ar@{-}[r] & \bullet_{mn-1}  &\triangleright})=
 $$
 $$
 \xymatrix{
   \triangleleft   &  \bullet_{2n-1}\ar@{-}[r] & \circ_{2n}    & \cdots&\circ_{(m-1)n}  &\\
  \bullet_1 \ar@{-}[d]  &  \circ_{2n-2} & \bullet_{2n+1} \ar@{-}[d]   & \cdots&\bullet_{(m-1)n+1}  &\\
  \circ_2           & \bullet_{2n-3} \ar@{-}[u]  & \circ_{2n+2} &   \cdots&\circ_{(m-1)n+2} \ar@{-}[u]& \\
  \vdots              & \vdots  & \vdots& \ddots & \vdots& \oplus\\
  \bullet_{n-2} \ar@{-}[d]  & \circ_{n+1}  & \bullet_{3n-2} \ar@{-}[d]& \cdots&\bullet_{mn-2}\ar@{-}[d]  &\\
  \circ_{n-1}           & \bullet_{n} \ar@{-}[u]  & \circ_{3n-1}    & \cdots&\circ_{mn-1} & }
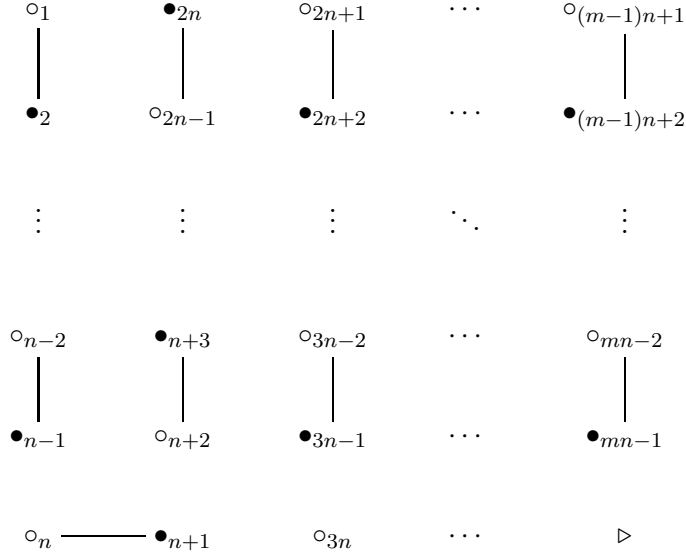
$$
\\
$$
 \xymatrix{
  \circ_1 \ar@{-}[d]  &  \bullet_{2n} & \circ_{2n+1} \ar@{-}[d]   & \cdots&\circ_{(m-1)n+1}  &\\
  \bullet_2           & \circ_{2n-1} \ar@{-}[u]  & \bullet_{2n+2} &   \cdots&\bullet_{(m-1)n+2} \ar@{-}[u]& \\
  \vdots              & \vdots  & \vdots& \ddots & \vdots&\\
  \circ_{n-2} \ar@{-}[d]  & \bullet_{n+3}  & \circ_{3n-2} \ar@{-}[d]& \cdots&\circ_{mn-2}\ar@{-}[d]  &\\
  \bullet_{n-1}           & \circ_{n+2} \ar@{-}[u]  & \bullet_{3n-1}    & \cdots&\bullet_{mn-1} & \\
   \circ_{n} \ar@{-}[r]          & \bullet_{n+1}  & \circ_{3n}    & \cdots&\triangleright &}
$$
Therefore, $i^{mn}_n(\mathfrak{h}_{mn}\oplus \mathfrak{h}_{mn}')=\tilde{\mathfrak{h}}_{mn}\oplus \tilde{\mathfrak{h}}_{mn}'$.

{\bf Case 3:} When $m$ is odd and $n$ is even,  $\mathcal{M}_{\phi}^{mn}$ keeps the same with the first case. Then $$(\mathcal{M}_{\phi}^{m})^n=((\mathfrak{u}^{mn},[\cdot,\cdot]_{\mathfrak{u}^{mn}},\phi_{(\mathfrak{u}^{m})^n},\langle\cdot,\cdot\rangle_{(\mathfrak{u}^{m})^n}),\tilde{\mathfrak{h}}_{mn},\tilde{\mathfrak{h}}_{mn}') $$
are given by
\begin{equation}\label{eq:mn-odd-even1}
\tilde{\mathfrak{h}}_{mn}=\underbrace{\mathfrak{u}_{\Delta}^m \oplus \cdots\oplus\mathfrak{u}_{\Delta}^m}_{\frac{n}{2}},
\end{equation}
\begin{equation}\label{eq:mn-odd-even2}
 \tilde{\mathfrak{h}}_{mn}' = \mathfrak{h}'\oplus\underbrace{\mathfrak{u}_{\Delta}\oplus\cdots
 \oplus\mathfrak{u}_{\Delta}}_{\frac{m-1}{2}}\oplus\underbrace{\mathfrak{u}_{\Delta}^m \oplus \cdots\oplus\mathfrak{u}_{\Delta}^m}_{\frac{n}{2}-1}\oplus
 \underbrace{\mathfrak{u}_{\Delta}\oplus\cdots\oplus\mathfrak{u}_{\Delta}}_{\frac{m-1}{2}}\oplus\mathfrak{h}.
\end{equation}
Then we have
$$
i^{mn}_n(\xymatrix{
  \circ_1 \ar@{-}[r] & \bullet_2  & \circ_3 \ar@{-}[r] & \bullet_4 &\cdots & \circ_{mn-1} \ar@{-}[r] & \bullet_{mn}  \quad\oplus}
 $$
 $$
 \xymatrix{
  \triangleleft &\bullet_1 \ar@{-}[r] & \circ_2  & \bullet_3 \ar@{-}[r] & \circ_4 &\cdots & \bullet_{mn-3} \ar@{-}[r] & \circ_{mn-2}  &\blacktriangleright})=
 $$
 $$
 \xymatrix{
  \circ_1 \ar@{-}[d]  &  \bullet_{2n} & \circ_{2n+1} \ar@{-}[d]   & \cdots&\circ_{(m-1)n+1}  &\\
  \bullet_2           & \circ_{2n-1} \ar@{-}[u]  & \bullet_{2n+2} &   \cdots&\bullet_{(m-1)n+2} \ar@{-}[u]& \\
  \vdots              & \vdots  & \vdots& \ddots & \vdots& \oplus\\
  \circ_{n-1}\ar@{-}[d]   & \bullet_{n+2}\ar@{-}[d]  & \circ_{3n-1}\ar@{-}[d] & \cdots&\circ_{mn-1}\ar@{-}[d]  &\\
  \bullet_n           & \circ_{n+1}   & \bullet_{3n}    & \cdots&\bullet_{mn} & }
 $$
\\
$$
 \xymatrix{
   \triangleleft   &  \bullet_{2n-1}\ar@{-}[r] & \circ_{2n}    & \cdots&\circ_{(m-1)n}  &\\
  \bullet_1 \ar@{-}[d]  &  \circ_{2n-2} & \bullet_{2n+1} \ar@{-}[d]   & \cdots&\bullet_{(m-1)n+1}  &\\
  \circ_2           & \bullet_{2n-3} \ar@{-}[u]  & \circ_{2n+2} &   \cdots&\circ_{(m-1)n+2} \ar@{-}[u]& \\
  \vdots              & \vdots  & \vdots& \ddots & \vdots&\\
  \bullet_{n-3} \ar@{-}[d]  & \circ_{n+2}  & \bullet_{3n-3} \ar@{-}[d]& \cdots&\bullet_{mn-3}\ar@{-}[d]  &\\
  \circ_{n-2}           & \bullet_{n+1} \ar@{-}[u]  & \circ_{3n-2}    & \cdots&\circ_{mn-2} & \\
   \bullet_{n-1} \ar@{-}[r]          & \circ_{n}  & \bullet_{3n-1}    & \cdots&\blacktriangleright &}
$$
Therefore, $i^{mn}_n(\mathfrak{h}_{mn}\oplus \mathfrak{h}_{mn}')=\tilde{\mathfrak{h}}_{mn}\oplus \tilde{\mathfrak{h}}_{mn}'$.

{\bf Case 4:} When $n$ is odd and $m$ is even, $\mathcal{M}_{\phi}^{mn}$ keeps the same with the first case. On the one hand, $$(\mathcal{M}_{\phi}^{m})^n=((\mathfrak{u}^{mn},[\cdot,\cdot]_{\mathfrak{u}^{mn}},\phi_{(\mathfrak{u}^{m})^n},\langle\cdot,\cdot\rangle_{(\mathfrak{u}^{m})^n}),\tilde{\mathfrak{h}}_{mn},\tilde{\mathfrak{h}}_{mn}') $$
are given by
\begin{align}
\tilde{\mathfrak{h}}_{mn}&=\underbrace{\mathfrak{u}_{\Delta}^m \oplus \cdots\oplus\mathfrak{u}_{\Delta}^m}_{\frac{n-1}{2}}\oplus
\underbrace{\mathfrak{u}_{\Delta}\oplus\cdots\oplus\mathfrak{u}_{\Delta}}_{\frac{m}{2}},\label{eq:mn-even-odd1}\\
\tilde{\mathfrak{h}}_{mn}' &= \mathfrak{h}'\oplus\underbrace{\mathfrak{u}_{\Delta}\oplus\cdots
\oplus\mathfrak{u}_{\Delta}}_{\frac{m}{2}-1}\oplus \mathfrak{h}\oplus\underbrace{\mathfrak{u}_{\Delta}^m \oplus \cdots\oplus\mathfrak{u}_{\Delta}^m}_{\frac{n-1}{2}},\label{eq:mn-even-odd2}
\end{align}
$\phi_{(\mathfrak{u}^{m})^n}$ and $\langle\cdot,\cdot\rangle_{(\mathfrak{u}^{m})^n}$ keep the same with the first case. On the other hand, we have
 $$
i^{mn}_n(\xymatrix{
  \circ_1 \ar@{-}[r] & \bullet_2  & \circ_3 \ar@{-}[r] & \bullet_4 &\cdots & \circ_{mn-1} \ar@{-}[r] & \bullet_{mn}  \quad\oplus}
 $$
 $$
 \xymatrix{
  \triangleleft &\bullet_1 \ar@{-}[r] & \circ_2  & \bullet_3 \ar@{-}[r] & \circ_4 &\cdots & \bullet_{mn-3} \ar@{-}[r] & \circ_{mn-2}  &\blacktriangleright})=
 $$
 $$
 \xymatrix{
  \circ_1 \ar@{-}[d]  &  \bullet_{2n} & \circ_{2n+1} \ar@{-}[d]   & \cdots&\bullet_{mn}  &\\
  \bullet_2           & \circ_{2n-1} \ar@{-}[u]  & \bullet_{2n+2} &   \cdots&\circ_{mn-1} \ar@{-}[u]& \\
  \vdots              & \vdots  & \vdots& \ddots & \vdots& \oplus\\
  \circ_{n-2} \ar@{-}[d]  &  \bullet_{n+3}\ar@{-}[d] & \circ_{3n-2} \ar@{-}[d]   & \cdots&\bullet_{(m-1)n+3}\ar@{-}[d]  &\\
  \bullet_{n-1}   & \circ_{n+2}  & \bullet_{3n-1} & \cdots&\circ_{(m-1)n+2}  &\\
  \circ_{n}\ar@{-}[r]           & \bullet_{n+1}   & \circ_{3n}    & \cdots&\bullet_{(m-1)n+1} & }
 $$
\\
\\
$$
 \xymatrix{
   \triangleleft   &  \bullet_{2n-1}\ar@{-}[r] & \circ_{2n}    & \cdots& \blacktriangleright  &\\
  \bullet_1 \ar@{-}[d]  &  \circ_{2n-2} & \bullet_{2n+1} \ar@{-}[d]   & \cdots&\circ_{mn-2}  &\\
  \circ_2           & \bullet_{2n-3} \ar@{-}[u]  & \circ_{2n+2} &   \cdots&\bullet_{mn-3} \ar@{-}[u]& \\
  \vdots              & \vdots  & \vdots& \ddots & \vdots&\\
  \bullet_{n-2}\ar@{-}[d]           & \circ_{n+1} \ar@{-}[d]  & \bullet_{3n-2}\ar@{-}[d]    & \cdots&\circ_{(m-1)n+1} & \\
   \circ_{n-1}        & \bullet_{n}  & \circ_{3n-1}    & \cdots& \bullet_{(m-1)n}\ar@{-}[u]&}
$$
Therefore, $i^{mn}_n(\mathfrak{h}_{mn}\oplus \mathfrak{h}_{mn}')=\tilde{\mathfrak{h}}_{mn}\oplus \tilde{\mathfrak{h}}_{mn}'$.

 Moreover, we have
$$\langle(a_1,\cdots,a_{mn}),(a_1',\cdots,a_{mn}')\rangle_{(\mathfrak{u}^{m})^n}=\sum_{k=0}^{n}\sum_{j=km+1}^{(k+1)m}(-1)^{k+j-1}\langle a_j,a_j'\rangle_{\mathfrak{u}}=\sum_{j=1}^{mn}(-1)^{j+1}\langle a_j,a_j'\rangle_{\mathfrak{u}}$$
$$
=\langle i^{mn}_n(a_1,\cdots,a_{mn}),i^{mn}_n(a_1',\cdots,a_{mn}')\rangle_{\mathfrak{u}^{mn}},
$$
and
$$
i^{mn}_n(\phi_{\mathfrak{u}^{mn}}(a_1,\cdots,a_{mn}))=i^{mn}_n(\phi_{\mathfrak{u}}(a_1),\cdots,\phi_{\mathfrak{u}}(a_{mn}))
=\phi_{\mathfrak{u}^{mn}}(i^{mn}_n(a_1,\cdots,a_{mn})).
$$
Note that $i^{mn}_n$ is unique because two adjacent vertices are different colors. Then $i^{mn}_n$  induced the isomorphism between  $\mathcal{M}_{\phi}^{mn}$ and $(\mathcal{M}_{\phi}^{m})^n$.
This finishes the proof.
\end{proof}

\begin{remark}
In Theorem \ref{tm:hom-polyubles} and Theorem \ref{tm:mn-uble}, restricting to the Lie cases $\mathfrak{h}=\fg_1,\mathfrak{h}'=\fg_2$ and $\fu=\fd$ with $\phi_{\mathfrak{h}}=\phi_{\mathfrak{h}'}=Id$, i.e. $\mathcal{M}_{\phi}=\mathcal{M}$. Then we obtain a  unique isomorphism as a Manin triple of Lie algebras $i^{mn}_n:\mathcal{M}^{nm}\rightarrow (\mathcal{M}^n)^m$, for any $m,n\in \mathbb{N}^*$.
\end{remark}

\section{Quasi-triangular hom-r-matrices, stabilizer hom-Lie subalgebras and Hom-Poisson structures}

\subsection{Quasi-triangular hom-r-matrices and hom-Schouten bracket}

\begin{definition}{\rm(\cite{CaiLiuSheng})}
Let $(\mathfrak{h},[\cdot,\cdot]_\mathfrak{h},\phi_\mathfrak{h})$ be a hom-Lie algebra. The {\bf hom-Schouten bracket} is a bilinear operation $[[\cdot,\cdot]]:\wedge^p\mathfrak{h}\otimes\wedge^q\mathfrak{h}\rightarrow \wedge^{p+q-1}\mathfrak{h}$ such that the following properties are satisfied.
\begin{itemize}
  \item
  The restriction of the hom-Schouten bracket to $\mathfrak{h}$ is a hom-Lie bracket, i.e.
    for all $X,Y,Z \in \mathfrak{h}$, the bracket $[[\cdot,\cdot]]:\wedge^2\mathfrak{h} \rightarrow \mathfrak{h}$ satisfies the hom-Jacobi identity:
    \begin{equation}
      [[\phi_\mathfrak{h}(X),[[Y,Z]]~]]+[[\phi_\mathfrak{h}(Y),[[Z,X]]~]]+
 [[\phi_\mathfrak{h}(Z),[[X,Y]]~]]=0.
    \end{equation}
  \item
  The bracket between two arbitrary elements is obtained according to the following two rules, for all~$X\in \wedge^p\mathfrak{h},~Y\in \wedge^q\mathfrak{h},~Z\in \wedge^l\mathfrak{h}$,
  \begin{align}
    [[X,Y]] & =-(-1)^{(p-1)(q-1)}[[Y,X]], \label{eq:hom-Schouten-bracket-1}\\
    [[X,Y\wedge Z]] & =[[X,Y]]\wedge \phi_\mathfrak{h}^{\otimes^l}(Z)+(-1)^{(p-1)l}\phi_\mathfrak{h}^{\otimes^q}(Y)\wedge[[X,Z]]\label{eq:hom-Schouten-bracket-2}.
  \end{align}
\end{itemize}
\end{definition}

\begin{lemma}{\rm(\cite{Teles})}
For all $X\in \wedge^i\mathfrak{h}, Y\in \wedge^j\mathfrak{h}, Z\in \wedge^k\mathfrak{h}$, the  hom-Schouten bracket determined by \eqref{eq:hom-Schouten-bracket-1} and \eqref{eq:hom-Schouten-bracket-2} satisfies the {\bf graded hom-Jacobi identity}:
\begin{align}
 &(-1)^{(i-1)(k-1)}[[\phi_\mathfrak{h}^{\otimes^i}(X),[[Y,Z]]~]]+(-1)^{(j-1)(i-1)}[[\phi_\mathfrak{h}^{\otimes^j}(Y),[[Z,X]]~]]\label{eq:ghJ}\\
  &+(-1)^{(k-1)(j-1)}[[\phi_\mathfrak{h}^{\otimes^k}(Z),[[X,Y]]~]]=0\nonumber.
\end{align}
\end{lemma}

\begin{remark}
The  hom-Schouten bracket $[[\cdot,\cdot]]:\wedge^2\mathfrak{h} \rightarrow \mathfrak{h}$ with $\phi_\mathfrak{h}= \mathrm{Id}$ is exactly the Schouten bracket in {\rm\cite{Dufour}}.
\end{remark}

Let $\mathfrak{h}^*$ be the dual vector space of $\mathfrak{h}$.
For each positive integer $p$, the $p$-linear pairing $\langle\cdot,\cdot\rangle:\overset{p}\otimes~\mathfrak{h}\times\overset{p}\otimes~\mathfrak{h}^*\rightarrow \mathbf{k}$ is defined by
\begin{equation}\label{eq:tensor-pair}
  \langle x_1\otimes\cdots\otimes x_p,\xi_1\otimes\cdots\otimes\xi_p\rangle=\langle x_1,\xi_1\rangle\langle x_2,\xi_2\rangle\cdots\langle x_p,\xi_p\rangle,\quad x_i\in \mathfrak{h},~\xi_i \in \mathfrak{h}^*.
\end{equation}

Let $(\mathfrak{h},[\cdot,\cdot]_\mathfrak{h},\phi_\mathfrak{h})$ be an involutive hom-Lie algebra,  $(\mathfrak{h},\mathrm{ad},\phi_{\mathfrak{h}})$ be the adjoint representation of $(\mathfrak{h},[\cdot,\cdot]_\mathfrak{h},\phi_\mathfrak{h})$ and
$(\mathfrak{h}^*,\mathrm{ad}^*,\phi_{\mathfrak{h}}^*)$ be the coadjoint representation of $(\mathfrak{h},[\cdot,\cdot]_\mathfrak{h},\phi_\mathfrak{h})$.
By \eqref{eq:tensor-pair}, for all $x\in \mathfrak{h}$, we have
\begin{equation}\label{eq:ad-tensor-pair}
  \langle \mathrm{ad}_x(x_1\otimes\cdots\otimes x_p),\xi_1\otimes\cdots\otimes\xi_p\rangle=-\langle x_1\otimes\cdots\otimes x_p,\sum_i^p\phi_\mathfrak{h}^*(\xi_1)\otimes\cdots\otimes\mathrm{ad}^*_x(\xi_i)\otimes\cdots\otimes\phi_\mathfrak{h}^*(\xi_p)\rangle.
\end{equation}

\begin{definition}
Let $(\mathfrak{h},[\cdot,\cdot]_\mathfrak{h},\phi_\mathfrak{h})$ be a hom-Lie algebra
and $r=\sum_{i}x_i\otimes y_i\in \mathfrak{h}\otimes\mathfrak{h}$. Define a  linear map $HCYB:\mathfrak{h}^{\otimes^2} \rightarrow\mathfrak{h}^{\otimes^3}$ by
\begin{align}\label{eq:hcyb}
  HCYB(r)  := &\sum_{i,j}([x_{i},x_j]_{\mathfrak{h}}\otimes \phi_{\mathfrak{h}}(y_i)\otimes \phi_{\mathfrak{h}}(y_j)+\phi_{\mathfrak{h}}(x_{i})\otimes [y_i,x_j]_{\mathfrak{h}}\otimes \phi_{\mathfrak{h}}(y_j)\nonumber\\
   &+\phi_{\mathfrak{h}}(x_{i})\otimes\phi_{\mathfrak{h}}(x_{j})\otimes[y_i,y_j]_{\mathfrak{h}}).
\end{align}
We call $HCYB$ the {\bf hom-classical Yang-Baxter map}.
\end{definition}

From now on, we suppose that the hom-Lie algebra $(\mathfrak{h},[\cdot,\cdot]_\mathfrak{h},\phi_\mathfrak{h})$ is involutive and $\phi_{\mathfrak{h}}^{\otimes 2}(r)=r$.

Let  $r^{21}=\sum_{i}y_i\otimes x_i$. From the facts
$\sum_{i}\phi_\mathfrak{h}(x_i)\otimes \phi_\mathfrak{h}(y_i)=\sum_{i}x_i\otimes y_i$ and $\phi_\mathfrak{h}^2=\mathrm{Id}$ we have
\begin{equation}\label{eq:phi-r}
  \sum_{i}\phi_\mathfrak{h}(x_i)\otimes y_i=\sum_{i}x_i\otimes \phi_\mathfrak{h}(y_i).
\end{equation}

Let $(\mathfrak{h}^*,\mathrm{ad}^*,\phi_{\mathfrak{h}}^*)$ be the coadjoint representation of $(\mathfrak{h},[\cdot,\cdot]_\mathfrak{h},\phi_\mathfrak{h})$. Where $\phi_\mathfrak{h}^*:\mathfrak{h}^*\rightarrow \mathfrak{h}^*$ is a linear map satisfies
\begin{equation}\label{eq:phih*}
  \langle\phi_\mathfrak{h}(x),\xi\rangle=\langle x,\phi_\mathfrak{h}^*(\xi)\rangle,\quad \forall~ x\in \mathfrak{h}.
\end{equation}
Define a linear map $r^\sharp: \mathfrak{h}^*\rightarrow  \mathfrak{h}$ satisfies
\begin{equation}\label{eq:r*}
  r^\sharp(\xi)=\sum_{i}\langle\phi_\mathfrak{h}^*(\xi),x_i\rangle y_i, \quad \xi \in~\mathfrak{h}^*,
\end{equation}

\begin{lemma}\label{lem:HCYBpair}
With the above notations, let
$r_+=r^\sharp:\mathfrak{h}^*\rightarrow \mathfrak{h}$ and $r_-=-(r^{21})^\sharp: \mathfrak{h}^*\rightarrow \mathfrak{h}$. Then we have, for all $\xi,\eta,\zeta\in \mathfrak{h}^*$,
\begin{equation}\label{eq:HCYBpair}
\langle HCYB(r), \xi\otimes\eta\otimes\zeta\rangle=\langle\xi,[r_-(\eta),r_-(\zeta)]_\mathfrak{h}\rangle
+\langle\eta,[r_-(\zeta),r_+(\xi)]_\mathfrak{h}\rangle+ \langle\zeta,[r_+(\xi),r_+(\eta)]_\mathfrak{h}\rangle.
\end{equation}
\end{lemma}
\begin{proof}
By \eqref{eq:hcyb}, the definition of $HCYB$, we have
\begin{eqnarray*}
\sum_{i,j}\langle[x_{i},x_j]_{\mathfrak{h}}\otimes \phi_{\mathfrak{h}}(y_i)\otimes \phi_{\mathfrak{h}}(y_j),\xi\otimes\eta\otimes\zeta\rangle
&\overset{\eqref{eq:tensor-pair}}{=}& \sum_{i,j}\langle[x_{i},x_j]_{\mathfrak{h}},\xi\rangle \langle\phi_{\mathfrak{h}}(y_i),\eta\rangle\langle \phi_{\mathfrak{h}}(y_j),\zeta\rangle\\
&\overset{\eqref{eq:phih*}}{=}& \sum_{i,j}\langle[x_{i},x_j]_{\mathfrak{h}},\xi\rangle \langle y_i,\phi_{\mathfrak{h}}^*(\eta)\rangle\langle y_j,\phi_{\mathfrak{h}}^*(\zeta)\rangle\\
&=&\langle\xi,[\sum_{i}\langle\phi_{\mathfrak{h}}^*(\eta),y_i\rangle x_i, \sum_{j}\langle\phi_{\mathfrak{h}}^*(\zeta),y_j\rangle x_j]_\mathfrak{h}\rangle\\
&=&\langle\xi,[r_-(\eta),r_-(\zeta)]_\mathfrak{h}\rangle,
\end{eqnarray*}
\begin{eqnarray*}
\sum_{i,j}\langle\phi_{\mathfrak{h}}(x_{i})\otimes [y_i,x_j]_{\mathfrak{h}}\otimes \phi_{\mathfrak{h}}(y_j),\xi\otimes\eta\otimes\zeta\rangle
  &=&\sum_{i,j}\langle\phi_{\mathfrak{h}}(x_{i}),\xi\rangle\langle[y_i,x_j]_{\mathfrak{h}},\eta\rangle \langle\phi_{\mathfrak{h}}(y_j),\zeta\rangle\\
  &=&\sum_{i,j}\langle x_{i},\phi_{\mathfrak{h}}^*(\xi)\rangle\langle[y_i,x_j]_{\mathfrak{h}},\eta\rangle \langle y_j,\phi_{\mathfrak{h}}^*(\zeta)\rangle\\
  &=&-\langle\eta,[\sum_{j}\langle\phi_{\mathfrak{h}}^*(\zeta),y_j\rangle x_j,\sum_{i}\langle\phi_{\mathfrak{h}}^*(\xi),x_i\rangle y_i]_\mathfrak{h}\rangle\\
  &=&\langle\eta,[r_-(\zeta),r_+(\xi)]_\mathfrak{h}\rangle,
\end{eqnarray*}
\begin{eqnarray*}
\sum_{i,j}\langle\phi_{\mathfrak{h}}(x_{i})\otimes\phi_{\mathfrak{h}}(x_{j})\otimes[y_i,y_j]_{\mathfrak{h}},
  \xi\otimes\eta\otimes\zeta\rangle
  &=&\sum_{i,j}\langle\phi_{\mathfrak{h}}(x_{i}),\xi\rangle\langle\phi_{\mathfrak{h}}(x_{j}),\eta\rangle \langle[y_i,y_j]_{\mathfrak{h}},\zeta\rangle\\
  &=&\sum_{i,j}\langle x_{i},\phi_{\mathfrak{h}}^*(\xi)\rangle\langle x_{j},\phi_{\mathfrak{h}}^*(\eta)\rangle \langle[y_i,y_j]_{\mathfrak{h}},\zeta\rangle\\
  &=&\langle\zeta,[\sum_{i}\langle\phi_{\mathfrak{h}}^*(\xi),x_i\rangle y_i,\sum_{j}\langle\phi_{\mathfrak{h}}^*(\eta),x_j\rangle y_j]_\mathfrak{h}\rangle\\
  &=&\langle\zeta,[r_+(\xi),r_+(\eta)]_\mathfrak{h}\rangle.
\end{eqnarray*}
Therefore, \eqref{eq:HCYBpair} holds.
\end{proof}
Denote $\Lambda$ and $S$ be the skew-symmetric and the symmetric parts of $r$, i.e., $r=\Lambda+S$, where $\Lambda\in\wedge^2\mathfrak{h}$, $S\in S^2\mathfrak{h}$.

\begin{proposition}
Let $\Lambda\in\wedge^2\mathfrak{h}$ and $\phi_{\mathfrak{h}}^{\otimes 2}(\Lambda)=\Lambda$. Then
\begin{equation}\label{eq:HCYBa}
HCYB(\Lambda)=\frac{1}{2}[[\Lambda,\Lambda]]\in \wedge^3 \mathfrak{h}.
\end{equation}
\end{proposition}
\begin{proof}
By Lemma \ref{lem:HCYBpair}, when $r=\Lambda$ is skew-symmetric, let $\Lambda=\sum_i^nx_i\wedge y_i$, $\Lambda^\sharp:\mathfrak{h}^*\rightarrow \mathfrak{h}$, $r_+=r_-=\Lambda^\sharp$. Then
\begin{equation}\label{eq:Lambda-sharp}
\Lambda^\sharp(\xi)=\sum_{i}\langle\phi_\mathfrak{h}^*(\xi),x_i\rangle y_i-\langle\phi_\mathfrak{h}^*(\xi),y_i\rangle x_i, \quad \forall~\xi \in~\mathfrak{h}^*.
\end{equation}
For all $\xi,\eta,\zeta\in \mathfrak{h}^*$, on the one hand, we have
\begin{eqnarray*}
&&\langle HCYB(\Lambda), \xi\otimes\eta\otimes\zeta\rangle\\
&=&\langle\xi,[\Lambda^\sharp(\eta),\Lambda^\sharp(\zeta)]_\mathfrak{h}\rangle
+\langle\eta,[\Lambda^\sharp(\zeta),\Lambda^\sharp(\xi)]_\mathfrak{h}\rangle+ \langle\zeta,[\Lambda^\sharp(\xi),\Lambda^\sharp(\eta)]_\mathfrak{h}\rangle\\
&=&\sum_{i,j}\langle\xi,[\langle\phi_\mathfrak{h}^*(\eta),x_i\rangle y_i-\langle\phi_\mathfrak{h}^*(\eta),y_i\rangle x_i,\langle\phi_\mathfrak{h}^*(\zeta),x_j\rangle y_j-\langle\phi_\mathfrak{h}^*(\zeta),y_j\rangle x_j]_\mathfrak{h}\rangle\\
&&+\sum_{i,j}\langle\eta,[\langle\phi_\mathfrak{h}^*(\zeta),x_i\rangle y_i-\langle\phi_\mathfrak{h}^*(\zeta),y_i\rangle x_i,\langle\phi_\mathfrak{h}^*(\xi),x_j\rangle y_j-\langle\phi_\mathfrak{h}^*(\xi),y_j\rangle x_j]_\mathfrak{h}\rangle\\
&&+\sum_{i,j}\langle\zeta,[\langle\phi_\mathfrak{h}^*(\xi),x_i\rangle y_i-\langle\phi_\mathfrak{h}^*(\xi),y_i\rangle x_i,\langle\phi_\mathfrak{h}^*(\eta),x_j\rangle y_j-\langle\phi_\mathfrak{h}^*(\eta),y_j\rangle x_j]_\mathfrak{h}\rangle\\
&=&\underset{\xi,\eta,\zeta}{\circlearrowleft}\sum_{i,j}\langle\xi,[\langle\eta,\phi_\mathfrak{h}(x_i)\rangle y_i-\langle\eta,\phi_\mathfrak{h}(y_i)\rangle x_i,\langle\zeta,\phi_\mathfrak{h}(x_j)\rangle y_j-\langle\zeta,\phi_\mathfrak{h}(y_j)\rangle x_j]_\mathfrak{h}\rangle\\
&=&\underset{\xi,\eta,\zeta}{\circlearrowleft}\sum_{i,j}\bigg(\langle\eta,\phi_\mathfrak{h}(x_i)\rangle\langle\zeta,\phi_\mathfrak{h}(x_j)\rangle\langle\xi,[ y_i,y_j]_\mathfrak{h}\rangle
-\langle\eta,\phi_\mathfrak{h}(x_i)\rangle\langle\zeta,\phi_\mathfrak{h}(y_j)\rangle\langle\xi,[y_i,x_j]_\mathfrak{h}\rangle\\
&&-\langle\eta,\phi_\mathfrak{h}(y_i)\rangle\langle\zeta,\phi_\mathfrak{h}(x_j)\rangle\langle\xi,[x_i,y_j]_\mathfrak{h}\rangle
+\langle\eta,\phi_\mathfrak{h}(x_j)\rangle\langle\zeta,\phi_\mathfrak{h}(y_j)\rangle\langle\xi,[x_i,x_j]_\mathfrak{h}\rangle\bigg)\\
&=&\underset{\xi,\eta,\zeta}{\circlearrowleft}\sum_{ij}\bigg(\langle\phi_\mathfrak{h}(x_i)\otimes\phi_\mathfrak{h}(x_j)\otimes[ y_i,y_j]_\mathfrak{h},\eta\otimes\zeta\otimes\xi\rangle
-\langle\phi_\mathfrak{h}(x_i)\otimes\phi_\mathfrak{h}(y_j)\otimes[y_i,x_j]_\mathfrak{h},\eta\otimes\zeta\otimes\xi\rangle\\
&&-\langle\phi_\mathfrak{h}(y_i)\otimes\phi_\mathfrak{h}(x_j)\otimes[x_i,y_j]_\mathfrak{h},\eta\otimes\zeta\otimes\xi\rangle
+\langle\phi_\mathfrak{h}(x_j)\otimes\phi_\mathfrak{h}(y_j)\otimes[x_i,x_j]_\mathfrak{h},\eta\otimes\zeta\otimes\xi\rangle\bigg).
\end{eqnarray*}
On the other hand,  by the definition of hom-Schouten bracket we have
\begin{eqnarray*}
[[\Lambda,\Lambda]]&=&\sum_{i,j}^n\bigg([x_i\wedge y_i,x_j]_\mathfrak{h}\wedge \phi (y_j)-\phi (x_j)\wedge[x_i\wedge y_i,y_j]_\mathfrak{h}\bigg)\\
&=&\sum_{i,j}^n\bigg(-[x_j,x_i\wedge y_i]_\mathfrak{h}\wedge \phi (y_j)+\phi (x_j)\wedge[y_j,x_i\wedge y_i]_\mathfrak{h}\bigg)\\
&=&\sum_{i,j}^n\bigg(-[x_j,x_i]_\mathfrak{h}\wedge \phi (y_i)\wedge \phi (y_j)-\phi (x_i)\wedge[x_j,y_i]_\mathfrak{h}\wedge \phi (y_j)\\
&&+\phi (x_j)\wedge[y_j,x_i]_\mathfrak{h}\wedge \phi (y_i)+\phi (x_j)\wedge \phi (x_i)\wedge [y_j,y_i]_\mathfrak{h}\bigg)\\
&=&\sum_{i,j}^n\bigg(\phi (y_i)\wedge \phi (y_j)\wedge[x_i,x_j]_\mathfrak{h}-\phi (x_i)\wedge \phi (y_j)\wedge [y_i,x_j]_\mathfrak{h}\\
&&-\phi (y_i)\wedge \phi (x_j)\wedge[x_i,y_j]_\mathfrak{h} +\phi (x_i)\wedge \phi (x_j)\wedge [y_i,y_j]_\mathfrak{h}\bigg),
\end{eqnarray*}
and
\begin{eqnarray*}
&&\langle[[\Lambda,\Lambda]],\xi\otimes\eta\otimes\zeta\rangle\\
&=&\sum_{i,j}^n\bigg(\langle \phi (y_i)\wedge \phi (y_j)\wedge[x_i,x_j]_\mathfrak{h},\xi\otimes\eta\otimes\zeta\rangle
-\langle \phi (x_i)\wedge \phi (y_j)\wedge [y_i,x_j]_\mathfrak{h},\xi\otimes\eta\otimes\zeta\rangle\\
&&-\langle \phi (y_i)\wedge \phi (x_j)\wedge[x_i,y_j]_\mathfrak{h},\xi\otimes\eta\otimes\zeta\rangle
+\langle \phi (x_i)\wedge \phi (x_j)\wedge [y_i,y_j],\xi\otimes\eta\otimes\zeta\rangle\bigg)\\
&=&\underset{\xi,\eta,\zeta}{\circlearrowleft}\sum_{i,j}^n\bigg(\langle \phi (y_i)\otimes \phi (y_j)\otimes[x_i,x_j]_\mathfrak{h},\xi\otimes\eta\otimes\zeta\rangle
-\langle \phi (y_j)\otimes \phi (y_i)\otimes[x_i,x_j]_\mathfrak{h},\xi\otimes\eta\otimes\zeta\rangle\\
&&-\langle \phi (x_i)\otimes \phi (y_j)\otimes [y_i,x_j]_\mathfrak{h},\xi\otimes\eta\otimes\zeta\rangle
+\langle \phi (y_j)\otimes \phi (x_i)\otimes [y_i,x_j]_\mathfrak{h},\xi\otimes\eta\otimes\zeta\rangle\\
&&-\langle \phi (y_j)\otimes \phi (x_i)\otimes[x_j,y_i]_\mathfrak{h},\xi\otimes\eta\otimes\zeta\rangle
+\langle \phi (x_i)\otimes \phi (y_j)\otimes[x_j,y_i]_\mathfrak{h},\xi\otimes\eta\otimes\zeta\rangle\\
&&+\langle \phi (x_i)\otimes \phi (x_j)\otimes [y_i,y_j]_\mathfrak{h},\xi\otimes\eta\otimes\zeta\rangle
-\langle \phi (x_j)\otimes \phi (x_i)\otimes [y_i,y_j]_\mathfrak{h},\xi\otimes\eta\otimes\zeta\rangle\bigg)\\
&=&2\underset{\xi,\eta,\zeta}{\circlearrowleft}\sum_{i,j}^n\bigg(\langle \phi (y_i)\otimes \phi (y_j)\otimes[x_i,x_j]_\mathfrak{h},\xi\otimes\eta\otimes\zeta\rangle
-\langle \phi (x_i)\otimes \phi (y_j)\otimes [y_i,x_j]_\mathfrak{h},\xi\otimes\eta\otimes\zeta\rangle\\
&&+\langle \phi (y_j)\otimes \phi (x_i)\otimes [y_i,x_j]_\mathfrak{h},\xi\otimes\eta\otimes\zeta\rangle
-\langle \phi (x_j)\otimes \phi (x_i)\otimes [y_i,y_j]_\mathfrak{h},\xi\otimes\eta\otimes\zeta\rangle\bigg).
\end{eqnarray*}
Therefore,
\begin{equation*}
HCYB(\Lambda)=\frac{1}{2}[[\Lambda,\Lambda]]\in \wedge^3 \mathfrak{h}.
\end{equation*}
\end{proof}

When $r=S$ is symmetric and satisfies $\phi_{\mathfrak{h}}^{\otimes 2}(S)=S$. Let $S=\sum_i^nx_i\otimes y_i$. Then  $S^\sharp:\mathfrak{h}^*\rightarrow \mathfrak{h}$, $r_+=S^\sharp$, $r_-=-S^\sharp$ and
\begin{equation}\label{eq:S-sharp}
S^\sharp(\xi)=\sum_{i}\langle\phi_\mathfrak{h}^*(\xi),x_i\rangle y_i,\quad \forall~\xi \in~\mathfrak{h}^*.
\end{equation}

\begin{proposition}\label{prop:hcybS}
With the above notations, let $S$ be hom-ad-invariant, i.e.,
$$
\langle\sum_i([x,x_i]_\mathfrak{h}\otimes \phi(y_i)+\phi(x_i)\otimes[x,y_i]_\mathfrak{h}),\xi\otimes\eta\rangle=0, \quad \forall ~ x\in\mathfrak{h},~~ \xi, \eta \in \mathfrak{h}^*.
$$
Then $HCYB(S)\in \wedge^3 \mathfrak{h}$ is hom-ad-invariant, and is given by:
\begin{equation}\label{eq:hcyb-s}
\langle HCYB(S),\xi\otimes\eta\otimes\zeta\rangle=\langle\zeta,[S^\sharp(\xi),S^\sharp(\eta)]_\mathfrak{h}\rangle,\quad \forall~\xi,\eta,\zeta\in \mathfrak{h}^*.
\end{equation}
\end{proposition}
\begin{proof}
Since $S$ is symmetric and hom-ad-invariant, we have
\begin{align*}
&\langle\sum_i [x,y_i]_\mathfrak{h}\otimes \phi_\mathfrak{h}(x_i),\xi\otimes\eta\rangle+\langle\sum_i \phi_\mathfrak{h}(x_i)\otimes[x,y_i]_\mathfrak{h}),\xi\otimes\eta\rangle\\
=&\sum_{i}\langle \eta,\phi_\mathfrak{h}(x_i)\rangle\langle\xi,[x, y_i]_\mathfrak{h} \rangle+\sum_{j}\langle \xi,\phi_\mathfrak{h}(x_j)\rangle\langle\eta,[x, y_j]_\mathfrak{h} \rangle\\
=&\langle\xi,[x,\sum_{i}\langle \phi_\mathfrak{h}^*(\eta),x_i\rangle y_i]_\mathfrak{h} \rangle+ \langle\eta,[x,\sum_{j}\langle \phi_\mathfrak{h}^*(\xi),x_j\rangle y_j]_\mathfrak{h} \rangle
=0,
\end{align*}
that is,
\begin{equation}\label{eq:symmetric}
  \langle\xi,[x,S^\sharp(\eta)]_\mathfrak{h} \rangle+ \langle\eta,[x,S^\sharp(\xi)]_\mathfrak{h} \rangle=0.
\end{equation}
By Lemma \ref{lem:HCYBpair}, we have
\begin{align*}
 \langle HCYB(S), \xi\otimes\eta\otimes\zeta\rangle&=\langle\xi,[S^\sharp(\eta),S^\sharp(\zeta)]_\mathfrak{h}\rangle
-\langle\eta,[S^\sharp(\zeta),S^\sharp(\xi)]_\mathfrak{h}\rangle+ \langle\zeta,[S^\sharp(\xi),S^\sharp(\eta)]_\mathfrak{h}\rangle\\
&=\langle\eta,[S^\sharp(\zeta),S^\sharp(\xi)]_\mathfrak{h}\rangle
-\langle\eta,[S^\sharp(\zeta),S^\sharp(\xi)]_\mathfrak{h}\rangle+ \langle\zeta,[S^\sharp(\xi),S^\sharp(\eta)]_\mathfrak{h}\rangle\\
&=\langle\zeta,[S^\sharp(\xi),S^\sharp(\eta)]_\mathfrak{h}\rangle.
\end{align*}
Which implies that
$$
\langle  HCYB(S),\xi\otimes\eta\otimes\zeta\rangle=-\langle  HCYB(S),\eta\otimes\xi\otimes\zeta\rangle=-\langle  HCYB(S),\zeta\otimes\eta\otimes\xi\rangle.
$$
Therefore, $HCYB(S)$ is skew-symmetric and \eqref{eq:hcyb-s} holds.

Moreover, on the one hand, we have
\begin{align*}
&\langle \xi,S^\sharp(\eta)\rangle=\langle \xi,\sum_{i}\langle\phi_\mathfrak{h}^*(\eta),x_i\rangle y_i\rangle=\langle \xi,\sum_{i}\langle\phi_\mathfrak{h}^*(\eta),y_i\rangle x_i\rangle=
\langle \xi,\sum_{i}\langle\eta,\phi(y_i)\rangle x_i\rangle=\sum_{i}\langle x_i\otimes \phi_\mathfrak{h} (y_i),\xi\otimes\eta\rangle.
\end{align*}
On the other hand,
$$
\langle S^\sharp(\xi),\eta\rangle=\langle \sum_{i}\langle\phi_\mathfrak{h}^*(\xi),x_i\rangle y_i,\eta\rangle=\langle \sum_{i}\langle\phi_\mathfrak{h}^*(\xi),y_i\rangle x_i,\eta\rangle=\langle \sum_{i}\langle\xi,\phi_\mathfrak{h}(y_i)\rangle x_i,\eta\rangle=\sum_{i}\langle \phi_\mathfrak{h} (y_i)\otimes x_i,\xi\otimes\eta\rangle.
$$
Then by \eqref{eq:phi-r} we have
\begin{equation}\label{eq:Ssharp-pair}
  \langle \xi,S^\sharp(\eta)\rangle=\langle S^\sharp(\xi),\eta\rangle.
\end{equation}

Since $(\mathfrak{h}^*,\mathrm{ad}^*,\phi_{\mathfrak{h}}^*)$ is the coadjoint representation of $(\mathfrak{h},[\cdot,\cdot]_\mathfrak{h},\phi_\mathfrak{h})$. By \eqref{eq:symmetric} and \eqref{eq:Ssharp-pair}, we have
\begin{equation}\label{eq:symmetricpair}
  \langle\xi,[x,S^\sharp(\eta)]_\mathfrak{h} \rangle- \langle S^\sharp\circ\mathrm{ad}_x^*(\eta),\xi \rangle=0,\quad \forall~ x\in\mathfrak{h},
\end{equation}
which implies that $\mathrm{ad}_x\circ S^\sharp=S^\sharp\circ\mathrm{ad}_x^*$.

By \eqref{eq:ad-tensor-pair} and the hom-Jacobi identity we have
\begin{eqnarray*}
  &&\langle[x,HCYB(S)]_{\mathfrak{h}},\xi\otimes\eta\otimes\zeta\rangle\\
   &=&-\langle HCYB(S),\mathrm{ad}_x^*(\xi)\otimes\phi_{\mathfrak{h}}^*(\eta)\otimes\phi_{\mathfrak{h}}^*(\zeta)\rangle
   -\langle HCYB(S),\phi_{\mathfrak{h}}^*(\xi)\otimes\mathrm{ad}_x^*(\eta)\otimes\phi_{\mathfrak{h}}^*(\zeta)\rangle\\
   &&-\langle HCYB(S),\phi_{\mathfrak{h}}^*(\xi)\otimes\phi_{\mathfrak{h}}^*(\eta)\otimes\mathrm{ad}_x^*(\zeta)\rangle\\
   &=&-\langle\phi_{\mathfrak{h}}^*(\zeta),[S^\sharp\mathrm{ad}_x^*(\xi),S^\sharp\phi_{\mathfrak{h}}^*(\eta)]_\mathfrak{h}\rangle-
   \langle\phi_{\mathfrak{h}}^*(\zeta),[S^\sharp\phi_{\mathfrak{h}}^*(\xi),S^\sharp\mathrm{ad}_x^*(\eta)]_\mathfrak{h}\rangle\\
   &&-\langle\mathrm{ad}_x^*(\zeta),[S^\sharp\phi_{\mathfrak{h}}^*(\xi),S^\sharp\phi_{\mathfrak{h}}^*(\eta)]_\mathfrak{h}\rangle\\
   &=&-\langle\phi_{\mathfrak{h}}^*(\zeta),[S^\sharp\mathrm{ad}_x^*(\xi),S^\sharp\phi_{\mathfrak{h}}^*(\eta)]_\mathfrak{h}\rangle-
   \langle\phi_{\mathfrak{h}}^*(\zeta),[S^\sharp\phi_{\mathfrak{h}}^*(\xi),S^\sharp\mathrm{ad}_x^*(\eta)]_\mathfrak{h}\rangle\\
   &&+\langle\zeta,[x,[S^\sharp\phi_{\mathfrak{h}}^*(\xi),S^\sharp\phi_{\mathfrak{h}}^*(\eta)]_\mathfrak{h}]_\mathfrak{h}\rangle\\
   &=&-\langle\zeta,\phi_{\mathfrak{h}}[\mathrm{ad}_x S^\sharp(\xi),S^\sharp\phi_{\mathfrak{h}}^*(\eta)]_\mathfrak{h}\rangle-
   \langle\zeta,\phi_{\mathfrak{h}}[S^\sharp\phi_{\mathfrak{h}}^*(\xi),\mathrm{ad}_x S^\sharp(\eta)]_\mathfrak{h}\rangle\\
   &&+\langle\zeta,[[\phi_{\mathfrak{h}}(x),S^\sharp\phi_{\mathfrak{h}}^*(\xi)]_\mathfrak{h},\phi_{\mathfrak{h}} S^\sharp\phi_{\mathfrak{h}}^*(\eta)]_\mathfrak{h}\rangle
   +\langle\zeta,[\phi_{\mathfrak{h}}S^\sharp\phi_{\mathfrak{h}}^*(\xi),[\phi_{\mathfrak{h}}(x), S^\sharp\phi_{\mathfrak{h}}^*(\eta)]_\mathfrak{h}\rangle\\
   &=&-\langle\zeta,[[\phi_{\mathfrak{h}}(x), \phi_{\mathfrak{h}} S^\sharp(\xi)]_\mathfrak{h},\phi_{\mathfrak{h}}S^\sharp\phi_{\mathfrak{h}}^*(\eta)]_\mathfrak{h}\rangle-
   \langle\zeta,[\phi_{\mathfrak{h}}S^\sharp\phi_{\mathfrak{h}}^*(\xi),[\phi_{\mathfrak{h}}(x) , \phi_{\mathfrak{h}}S^\sharp(\eta)]_\mathfrak{h}]_\mathfrak{h}\rangle\\
   &&+\langle\zeta,[[\phi_{\mathfrak{h}}(x),S^\sharp\phi_{\mathfrak{h}}^*(\xi)]_\mathfrak{h},\phi_{\mathfrak{h}} S^\sharp\phi_{\mathfrak{h}}^*(\eta)]_\mathfrak{h}\rangle
   +\langle\zeta,[\phi_{\mathfrak{h}}S^\sharp\phi_{\mathfrak{h}}^*(\xi),[\phi_{\mathfrak{h}}(x), S^\sharp\phi_{\mathfrak{h}}^*(\eta)]_\mathfrak{h}]_\mathfrak{h}\rangle.
\end{eqnarray*}
By $\phi_\mathfrak{h}^2=\mathrm{Id}$ and $\phi_{\mathfrak{h}}^{\otimes 2}(S)=S$, we have
$$
\phi_{\mathfrak{h}}\circ S^\sharp (\xi)= \sum_{i}\langle\phi_\mathfrak{h}^*(\xi),x_i\rangle \phi_{\mathfrak{h}}(y_i)=\sum_{i}\langle\phi_\mathfrak{h}^*(\xi),\phi_\mathfrak{h}(x_i)\rangle y_i
=\sum_{i}\langle\xi,x_i\rangle y_i=S^\sharp\circ \phi_{\mathfrak{h}}^*(\xi),
$$
then $\phi_{\mathfrak{h}}\circ S^\sharp=S^\sharp\circ \phi_{\mathfrak{h}}^*$. Therefore,
$$\langle[x,HCYB(S)]_{\mathfrak{h}},\xi\otimes\eta\otimes\zeta\rangle=0,$$
and $HCYB(S)$ is hom-ad-invariant, which completes the proof.
\end{proof}

\begin{theorem}\label{tem:HCYB}
Let  $S\in S^2 \mathfrak{h}$ be hom-ad-invariant, $\Lambda\in\wedge^2\mathfrak{h}$ and $\phi_{\mathfrak{h}}^{\otimes 2}(\Lambda+S)=\Lambda+S$. Then
\begin{equation}\label{eq:MHCYB}
HCYB(\Lambda+S)=HCYB(\Lambda)+HCYB(S).
\end{equation}
\end{theorem}
\begin{proof}
For $r=\Lambda+S$, by Lemma \ref{lem:HCYBpair} we have $r_+=\Lambda^\sharp+S^\sharp$ and $r_-=\Lambda^\sharp-S^\sharp$. Then
\begin{eqnarray*}
  &&\langle HCYB(S+\Lambda),\xi\otimes\eta\otimes\zeta\rangle\\
   &= &\langle\xi,[(\Lambda^\sharp-S^\sharp)(\eta),(\Lambda^\sharp-S^\sharp)(\zeta)]_\mathfrak{h}\rangle
+\langle\eta,[(\Lambda^\sharp-S^\sharp)(\zeta),(\Lambda^\sharp+S^\sharp)(\xi)]_\mathfrak{h}\rangle\\
&&+ \langle\zeta,[(\Lambda^\sharp+S^\sharp)(\xi),(\Lambda^\sharp+S^\sharp)(\eta)]_\mathfrak{h}\rangle\\
&= & \langle\xi,[\Lambda^\sharp(\eta),\Lambda^\sharp(\zeta)]-[\Lambda^\sharp(\eta),S^\sharp(\zeta)]
-[S^\sharp(\eta),\Lambda^\sharp(\zeta)+[S^\sharp(\eta),S^\sharp(\zeta)]\rangle\\
&&+\langle\eta,[\Lambda^\sharp(\zeta),\Lambda^\sharp(\xi)]+[\Lambda^\sharp(\zeta),S^\sharp(\xi)]
-[S^\sharp(\zeta),\Lambda^\sharp(\xi)]-[S^\sharp(\zeta),S^\sharp(\xi)]\rangle\\
&&+\langle\zeta,[\Lambda^\sharp(\xi),\Lambda^\sharp(\eta)]+[\Lambda^\sharp(\xi),S^\sharp(\eta)]
+[S^\sharp(\xi),\Lambda^\sharp(\eta)]+[S^\sharp(\xi),S^\sharp(\eta)]\rangle\\
&= & \langle HCYB(\Lambda),\xi\otimes\eta\otimes\zeta\rangle+\langle HCYB(S),\xi\otimes\eta\otimes\zeta\rangle\\
&&+\langle\xi,-[\Lambda^\sharp(\eta),S^\sharp(\zeta)]
-[S^\sharp(\eta),\Lambda^\sharp(\zeta)]\rangle
+\langle\eta,[\Lambda^\sharp(\zeta),S^\sharp(\xi)]
-[S^\sharp(\zeta),\Lambda^\sharp(\xi)]\rangle\\
&&+\langle\zeta,[\Lambda^\sharp(\xi),S^\sharp(\eta)]
+[S^\sharp(\xi),\Lambda^\sharp(\eta)]\rangle.
\end{eqnarray*}
By \eqref{eq:Lambda-sharp} and \eqref{eq:S-sharp}, the definition of $\Lambda^\sharp$ and $S^\sharp$, we have
\begin{eqnarray*}
&&\langle\xi,[\Lambda^\sharp(\eta),S^\sharp(\zeta)]_\mathfrak{h}\rangle\\
&=&\langle\xi,[\sum_{i}\langle\phi_\mathfrak{h}^*(\eta),x_i\rangle y_i-\langle\phi_\mathfrak{h}^*(\eta),y_i\rangle x_i,\sum_{j}\langle\phi_\mathfrak{h}^*(\zeta),x_j\rangle y_j]_\mathfrak{h}\rangle\\
&=&\sum_{i,j}\langle\phi_\mathfrak{h}^*(\eta),x_i\rangle \langle\phi_\mathfrak{h}^*(\zeta),x_j\rangle\langle\xi,[ y_i,y_j]_\mathfrak{h}\rangle-\sum_{i,j}\langle\phi_\mathfrak{h}^*(\eta),y_i\rangle\langle\phi_\mathfrak{h}^*(\zeta),x_j\rangle \langle\xi,[x_i,y_j]_\mathfrak{h}\rangle\\
&=&-\sum_{i,j}\langle\phi_\mathfrak{h}^*(\eta),x_i\rangle \langle\phi_\mathfrak{h}^*(\zeta),x_j\rangle\langle\xi,[ y_j,y_i]_\mathfrak{h}\rangle+\sum_{i,j}\langle\phi_\mathfrak{h}^*(\eta),y_i\rangle\langle\phi_\mathfrak{h}^*(\zeta),x_j\rangle \langle\xi,[y_j,x_i]_\mathfrak{h}\rangle\\
&=&-\langle\xi,[\sum_{j}\langle\phi_\mathfrak{h}^*(\zeta),x_j\rangle y_j,\sum_{i}\langle\phi_\mathfrak{h}^*(\eta),x_i\rangle y_i-\langle\phi_\mathfrak{h}^*(\eta),y_i\rangle x_i]_\mathfrak{h}\rangle\\
&=&-\langle\xi,[S^\sharp(\eta),\Lambda^\sharp(\zeta)]_\mathfrak{h}\rangle.
\end{eqnarray*}
Then
$$
\langle HCYB(S+\Lambda),\xi\otimes\eta\otimes\zeta\rangle=\langle HCYB(\Lambda),\xi\otimes\eta\otimes\zeta\rangle+\langle HCYB(S),\xi\otimes\eta\otimes\zeta\rangle,
$$
which proves the theorem.
\end{proof}

\begin{corollary}
Let $r=\Lambda+S \in \mathfrak{h}\otimes \mathfrak{h}$ with $\Lambda\in\wedge^2\mathfrak{h}$, $S\in S^2 \mathfrak{h}$. If $S$ is hom-ad-invariant and $\phi_{\mathfrak{h}}^{\otimes 2}(\Lambda+S)=\Lambda+S$. A sufficient condition for $HCYB(r)$ to be hom-ad-invariant is
\begin{equation}\label{eq:MHYBE}
HCYB(\Lambda)+HCYB(S)=0.
\end{equation}
The condition $HCYB(r)=0$ is called {\bf  hom-classical Yang-Baxter equation}, abbreviated by $HCYBE$.
\end{corollary}

\begin{definition}
Let $r=\Lambda+S$ be an element in $\mathfrak{h}\otimes\mathfrak{h}$, with the skew-symmetric part $\Lambda$ and symmetric part $S$, such that
\begin{equation}\label{eq:r-condition}
  \phi_{\mathfrak{h}}^{\otimes^2}(r)=r.
\end{equation}
 If $S$ is hom-ad-invariant and $r$  satisfies the $HCYBE$, then $r$ is called the {\bf quasi-triangular hom-$r$-matrix}. If, moreover, $S^\sharp$ is invertible, then $r$ is called {\bf factorizable}. If $r=\Lambda$ is skew-symmetric and satisfies the $HCYBE$, then $r$ is called the {\bf skew-symmetric hom-$r$-matrix}.
\end{definition}

\begin{remark}
The notion of quasi-triangular hom-$r$-matrix has already appeared in \cite{Yau1} and \cite{ShengBai}. Here we set $r=\Lambda+S$ and consider that $r$ has symmetric and antisymmetric parts.
If $\phi_{\mathfrak{h}}= \mathrm{Id}$, then the $HCYBE$ reduces to the $CYBE$, i.e. for all $r=\sum_i x_i\otimes y_i \in \mathfrak{h}\otimes \mathfrak{h}$,
\begin{align*}
  &CYB(r) = \sum_{i,j}([x_{i},x_j]\otimes y_i\otimes y_j+x_{i}\otimes [y_i,x_j]\otimes y_j
   +x_{i}\otimes x_{j}\otimes[y_i,y_j])=0.
\end{align*}
In this case, a solution of $HCYBE$ is just a classical $r$-matrix.
\end{remark}

\begin{example}\label{exam:sl2}
Let $(\mathfrak{h},[\cdot,\cdot]_\mathfrak{h},\phi_\mathfrak{h})$ be an involutive hom-Lie algebra of dimension $3$ with a basis $\{e_1,e_2,e_3\}$. Where $[\cdot,\cdot]_\mathfrak{h}$ is defined by
\begin{equation}\label{eq:sl2}
[e_1,e_2]_\mathfrak{h}=-2e_2,\quad[e_1,e_3]_\mathfrak{h}=2e_3,\quad[e_2,e_3]_\mathfrak{h}=e_1.
\end{equation}

The linear map $\phi_\mathfrak{h}$ is defined by
$$
\phi_\mathfrak{h}(e_1)=e_1,\quad \phi_\mathfrak{h}(e_2)=-e_2,\quad\phi_\mathfrak{h}(e_3)=-e_3.
$$
Set
$$
\Lambda=\frac{1}{2}(e_2\otimes e_3-e_3\otimes e_2),\quad S=\frac{1}{4}e_1\otimes e_1+\frac{1}{2}(e_2\otimes e_3+e_3\otimes e_2).
$$
Then
$$
r=\Lambda+S=e_2\otimes e_3+\frac{1}{4}e_1\otimes e_1
$$
is a quasi-triangular hom-$r$-matrix, $S$ is hom-ad-invariant and  $\phi_{\mathfrak{h}}^{\otimes^2}(r)=r$.
\end{example}

\begin{remark}
Let $\mathfrak{sl}(2,\mathbb{C})$  be a Lie algebra with a basis $\{e_1,e_2,e_3\}$ and the Lie bracket is given by \eqref{eq:sl2}. Set
$$
r=\Lambda+S=e_2\otimes e_3+\frac{1}{4}e_1\otimes e_1.
$$
Then $r$ is not a classical $r$-matrix, i.e.,
$$CYB(r)=-2e_2\otimes e_1\otimes e_3\neq 0.$$
\end{remark}

\subsection{Construction of hom-Poisson structures and proof of Theorem B(1)}

Let $M$ be a smooth manifold of dimension $n$, $TM$ be the tangent bundle and $\varphi:M\rightarrow M$ be a smooth map. Then the pullback map $\varphi^*:C^{\infty}(M)\rightarrow C^{\infty}(M)$ is a morphism of the function ring $C^{\infty}(M)$, i.e.,
\begin{align*}
\varphi^*(fg)&=\varphi^*((\varphi\circ \varphi^*)(f)(\varphi\circ \varphi^*)(g))\\
&=\varphi^*\circ\varphi(\varphi^*(f) \varphi^*(g))\\
&=\varphi^*(f)\varphi^*(g),\quad \forall~f,g\in C^{\infty}(M).
\end{align*}

\begin{lemma} {\rm(\cite{ShengZhen})}
Let $V$ be a vector space, and $\beta\in \mathrm{GL}(V)$. Define a skew-symmetric bilinear map $[\cdot,\cdot]_{\beta}:\wedge^2\mathfrak{gl}(V)\rightarrow \mathfrak{gl}(V)$ by
\begin{equation}\label{eq:rhl}
  [A,B]_\beta=\beta\circ A \circ \beta^{-1}\circ B\circ \beta^{-1}- \beta\circ B \circ \beta^{-1}\circ A\circ \beta^{-1},\quad \forall~A,B\in \mathfrak{gl}(V).
\end{equation}
Define the adjoint action  $Ad_{\beta}:\mathfrak{gl}(V)\rightarrow \mathfrak{gl}(V)$ by
\begin{equation}\label{eq:Adb}
 Ad_{\beta}(A)=\beta\circ A\circ  \beta^{-1}, \quad \forall~A\in \mathfrak{gl}(V).
\end{equation}
Then $(\mathfrak{gl}(V),[\cdot,\cdot]_\beta,Ad_{\beta})$ is a hom-Lie algebra.
\end{lemma}

\begin{lemma}{\rm(\cite{Teles})}
Let $\varphi^!TM$ be the pullback bundle of $TM$, and set
$$\Gamma(\varphi^!TM)=\{x:M\rightarrow TM~|~x=X\circ \varphi, ~\forall~X\in\Gamma(TM)\}.$$
Then we have
\begin{equation}\label{eq:Derphi}
  x(fg)=x(f)\varphi^*(g)+ \varphi^*(f)x(g), \quad\forall~f,g\in C^{\infty}(M),x\in \Gamma(\varphi^!TM).
\end{equation}
\end{lemma}

\begin{lemma}{\rm(\cite{CaiLiuSheng})}
Define a skew-symmetric bilinear operation $[\cdot,\cdot]_{\varphi^*}:\wedge^2\Gamma(\varphi^!TM) \rightarrow \Gamma(\varphi^!TM)$ by
$$
[x,y]_{\varphi^*}=\varphi^*\circ x \circ (\varphi^*)^{-1}\circ y \circ (\varphi^*)^{-1}
-\varphi^*\circ y \circ (\varphi^*)^{-1}\circ x \circ (\varphi^*)^{-1},\quad \forall~x,y\in \Gamma(\varphi^!TM).
$$
And define a linear map $\mathrm{Ad}_{\varphi^*}:\Gamma(\varphi^!TM)\rightarrow \Gamma(\varphi^!TM)$ by
$$
\mathrm{Ad}_{\varphi^*}(x)=\varphi^*\circ x\circ (\varphi^*)^{-1}, \quad \forall~x\in \Gamma(\varphi^!TM).
$$
Then $(\Gamma(\varphi^!TM),[\cdot,\cdot]_{\varphi^*},\mathrm{Ad}_{\varphi^*})$ is a hom-Lie algebra.
\end{lemma}

\begin{definition}{\rm(\cite{CaiLiuSheng})}\label{def:Hom-action}
Let $(\mathfrak{h},[\cdot,\cdot]_\mathfrak{h},\phi_\mathfrak{h})$ be a hom-Lie algebra and let $\varphi:M\rightarrow M$ be a diffeomorphism. An {\bf action} of $(\mathfrak{h},[\cdot,\cdot]_\mathfrak{h},\phi_\mathfrak{h})$ on $M$ is a linear map
$\sigma:\mathfrak{h}\rightarrow \Gamma(\varphi^!TM)$, such that, for all $x,y\in \mathfrak{h}$,
\begin{align}
&\sigma(\phi_\mathfrak{h}(x))=\mathrm{Ad}_{\varphi^*}(\sigma(x)), \label{eq:action1}\\
&\sigma([x,y]_{\mathfrak{h}})=[\sigma(x),\sigma(y)]_{\varphi^*}\label{eq:action2}.
\end{align}
\end{definition}

\begin{definition}\label{def:Hom-Poisson}
A  bisection $\pi\in \wedge^2 \Gamma(\varphi^!TM)$ on a smooth manifold $M$ is said to be the {\bf hom-Poisson
tensor} if the hom-Schouten bracket on $\Gamma(\varphi^!TM)$ satisfies $[[\pi,\pi]]_{\Gamma(\varphi^!TM)}=0$ and $\mathrm{Ad}_{\varphi^*}(\pi)=\pi$. A {\bf hom-Poisson structure} is a manifold $M$ equipped
with a hom-Poisson tensor $\pi$. We denote the hom-Poisson structure by $(M,\varphi,\pi)$.
\end{definition}
For a subspace $\mathfrak{q}_y$ of $\mathfrak{h}$, define the \textbf{stabilizer hom-subalgebra} of hom-Lie algebra $(\mathfrak{h},[\cdot,\cdot]_\mathfrak{h},\phi_\mathfrak{h})$ at $y\in M$ as
\begin{equation}\label{stabilizer-subalgebra}
  \mathfrak{q}_y=\{a\in \mathfrak{h}~|~\sigma(a)(y)=0\}\quad\mbox{and}\quad \phi_\mathfrak{h}(\mathfrak{q}_y)\in \mathfrak{q}_y.
\end{equation}
Set
\begin{equation}\label{eq:bot}
  \mathfrak{q}_y^{\bot}=\{\xi\in \mathfrak{h}^*~|~\langle \xi ,x\rangle=0,~\forall~x\in\mathfrak{q}_y\}.
\end{equation}
Let $r=\sum_{i}x_i\otimes y_i\in \mathfrak{h}\otimes \mathfrak{h},~x_i,y_i\in \mathfrak{h}$, and define
\begin{equation*}
  \sigma(r)=\sum_{i} \sigma(x_i)\otimes \sigma(y_i).
\end{equation*}
\begin{theorem}\label{Hom-Poisson}
Let  $r=\Lambda+S\in \mathfrak{h}\otimes \mathfrak{h}$ be a quasi-triangular hom-$r$-matrix, $S^{\sharp}(\mathfrak{q}_y^{\bot})\subset \mathfrak{q}_y$ and $\sigma:\mathfrak{h}\rightarrow \Gamma(\varphi^!TM)$ be an action of $(\mathfrak{h},[\cdot,\cdot]_\mathfrak{h},\phi_\mathfrak{h})$ on a manifold $M$. Then $(M,\varphi,\sigma(r))$ is a hom-Poisson structure.
\end{theorem}
\begin{proof}
We first prove that $\sigma(r)=\sigma(\Lambda)$, i.e. $\sigma(S)=0$. Let $r=\sum_{i}x_i\otimes y_i,~x_i,y_i\in \mathfrak{h}$. For all $\xi_1,\xi_2\in \mathfrak{q}_y^{\bot}$, by  $S^{\sharp}(\mathfrak{q}_y^{\bot})\subset \mathfrak{q}_y$ we have
\begin{align*}
\langle S^{\sharp}(\xi_1),\xi_2\rangle=\langle\sum_{i}\langle\phi_\mathfrak{h}^*(\xi_1),x_i\rangle y_i,\xi_2\rangle=0.
\end{align*}
Then $S\in \mathfrak{q}_y\otimes \mathfrak{h}$, which implies that $\sigma(S)=0$.

By  the the definition of hom-Schouten bracket and  $\sigma$, for all $\Lambda=\sum_{i}x_i\wedge y_i$, we have
\begin{align*}
[[\sigma(\Lambda),\sigma(\Lambda)]]_{\Gamma(\varphi^!TM)}=&\sum_{i,j}[[\sigma(x_i)\wedge \sigma(y_i),\sigma(x_j)\wedge \sigma(y_j)]]_{\Gamma(\varphi^!TM)}\\
=&\sum_{i,j}\bigg([\sigma(x_i),\sigma(x_j)]_{\varphi^*}\wedge \mathrm{Ad}_{\varphi^*}(\sigma(y_i))\wedge \mathrm{Ad}_{\varphi^*}(\sigma(y_j))\\
&-[\sigma(x_i),\sigma(y_j)]_{\varphi^*}\wedge \mathrm{Ad}_{\varphi^*}(\sigma(y_i))\wedge \mathrm{Ad}_{\varphi^*}(\sigma(x_j))\\
&-[\sigma(y_i),\sigma(x_j)]_{\varphi^*}\wedge \mathrm{Ad}_{\varphi^*}(\sigma(x_i))\wedge \mathrm{Ad}_{\varphi^*}(\sigma(y_j))\\
&+[\sigma(y_i),\sigma(y_j)]_{\varphi^*}\wedge \mathrm{Ad}_{\varphi^*}(\sigma(x_i))\wedge \mathrm{Ad}_{\varphi^*}(\sigma(x_j))
\bigg)\\
=&\sum_{i,j}\bigg(\sigma([x_i,x_j]_{\mathfrak{h}})\wedge \sigma(\phi_\mathfrak{h}(y_i))\wedge \sigma(\phi_\mathfrak{h}(y_j))\\
&-\sigma([x_i,y_j]_{\mathfrak{h}})\wedge \sigma(\phi_\mathfrak{h}(y_i))\wedge \sigma(\phi_\mathfrak{h}(x_j))\\
&-\sigma([y_i,x_j]_{\mathfrak{h}})\wedge \sigma(\phi_\mathfrak{h}(x_i))\wedge \sigma(\phi_\mathfrak{h}(y_j))\\
&+\sigma([y_i,y_j]_{\mathfrak{h}})\wedge \sigma(\phi_\mathfrak{h}(x_i))\wedge \sigma(\phi_\mathfrak{h}(x_j))
\bigg)\\
=&\sigma([[\Lambda,\Lambda]]).
\end{align*}

We next prove that
\begin{align*}
[[\sigma(r),\sigma(r)]]_{\Gamma(\varphi^!TM)}=[[\sigma(\Lambda),\sigma(\Lambda)]]_{\Gamma(\varphi^!TM)}=\sigma([[\Lambda,\Lambda]])=0.
\end{align*}
By $r$ is a quasi-triangular hom-$r$-matrix and Theorem \ref{tem:HCYB} we have
\begin{align*}
HCYB(\Lambda)+HCYB(S)= \frac{1}{2}[[\Lambda,\Lambda]]+HCYB(S)=0.
\end{align*}
Then by Proposition \ref{prop:hcybS} and $S^{\sharp}(\mathfrak{q}_y^{\bot})\subset \mathfrak{q}_y$, for all $\xi_1,\xi_2,\xi_3\in \mathfrak{q}_y^{\bot}$ we have
\begin{align*}
\langle[[\Lambda,\Lambda]],\xi_1\otimes\xi_2\otimes\xi_3\rangle&=-2\langle HCYB(S),\xi_1\otimes\xi_2\otimes\xi_3\rangle\\
&=-2\langle\xi_3,[S^\sharp(\xi_1),S^\sharp(\xi_2)]_\mathfrak{h}\rangle\\
&=0,
\end{align*}
which implies that
 $$[[\Lambda,\Lambda]]\in \mathfrak{q}_y\wedge \mathfrak{h}\wedge \mathfrak{h},\quad\forall~y\in M.$$
Therefore, $[[\sigma(r),\sigma(r)]]_{\Gamma(\varphi^!TM)}=0$.

 By \eqref{eq:action1}, the definition of $\sigma$, we have
\begin{align*}
\mathrm{Ad}_{\varphi^*}(\sigma(r))=\sigma(\phi_\mathfrak{h}^{\otimes 2}(r))=\sigma(r),
\end{align*}
which completes the proof.
\end{proof}
The following results can be proved in the same way as the above theorem.
\begin{corollary}\label{coro:Lambda+S}
Let $r=\Lambda+S\in \mathfrak{h}\otimes \mathfrak{h}$ be a quasi-triangular hom-$r$-matrix. Then $(M,\varphi,\sigma(\Lambda))$ is a hom-Poisson structure if and only if $$[S^\sharp(\mathfrak{q}_y^{\bot}),S^\sharp(\mathfrak{q}_y^{\bot})]_\mathfrak{h}\subset \mathfrak{q}_y.$$
\end{corollary}

\begin{corollary}\label{coro:Lambda}
Let $r=\Lambda$ be a skew-symmetric hom-$r$-matrix. Then $(M,\varphi,\sigma(\Lambda))$ is a hom-Poisson structure.
\end{corollary}

\begin{remark}\label{Poisson}
In particular, in Definition \ref{def:Hom-action} and Definition \ref{def:Hom-Poisson}, actions of hom-Lie algebras $(\mathfrak{h},[\cdot,\cdot]_\mathfrak{h},\phi_\mathfrak{h})$ with $\phi_\mathfrak{h}=Id$ on $\mathfrak{h}$ and hom-Poisson structures $(M,\varphi,\pi)$ with $\varphi=Id$ on $M$ are exactly the actions of Lie algebras and Poisson structures, respectively. Let $\mathfrak{h}= \mathfrak{g}$ be a Lie algebra and let $r=\Lambda+S\in \mathfrak{g}\otimes \mathfrak{g}$ be a quasi-triangular $r$-matrix. If $\mathfrak{q}_m$ is the stabilizer Lie subalgebra at $m\in M$, i.e. $\mathfrak{q}_{m}:=\{a\in \mathfrak{g}:\sigma(a)(m)=0\}$,   then by Corollary \ref{coro:Lambda+S} the bivector field $\sigma(\Lambda)$ is a Poisson structure if and only if $[\mathfrak{q}_{m}^{\perp},\mathfrak{q}_{m}^{\perp}]\subset \mathfrak{q}_{m}$.
\end{remark}
\section{Poisson homogeneous spaces associated to polyubles}
\subsection{Poisson Lie group and Manin triple}
\begin{definition}{\rm(\cite{PGLecture})}
Let $G$ be a Lie group and $\pi_G$ a Poisson structure. The pair $(G,\pi_G)$ is called {\rm \textbf{a Poisson Lie group}} if it satisfies one of following equivalent conditions:
\begin{itemize}
\item Equipping $G\times G$ with the product Poisson structure, the group multiplication map is a Poisson map.
\item $\pi_{G}(gh)=l_g(\pi_G(h))+r_h(\pi_G(g))$.
\end{itemize}
\end{definition}

\begin{definition}{\rm(\cite{PGLecture})}
 A pair $(\mathfrak{g},\nu)$, where $\mathfrak{g}$ is a Lie algebra and $\nu:\mathfrak{g}\rightarrow \wedge^2 \mathfrak{g}$ is a linear map, is called {\rm \textbf{a Lie bialgebra}} , if it satisfies that
\begin{itemize}
\item $\nu[x,y]=[x,\nu (y)]-[y,\nu (x)]$, $x,y\in \mathfrak{g}$;
\item the dual map of $\nu$: $\nu^*:\wedge^2 \g^*\rightarrow \fg^*$ is a Lie bracket on $\fg^*$.
\end{itemize}
\end{definition}

Let $G$ be a Lie group with a Poisson structure $\pi_G$. Define the linear map $d_e\pi_G: \fg\rightarrow \wedge^2 \fg$
by
\[(d_e\pi_G)(x)=(L_{\tilde{x}}\pi_{G})(e),\]
where for $x\in \g$, $\tilde{x}$ is a vector field of $G$ such that $\tilde{x}(e)=x$ and $L_{\tilde{x}}$ is the Lie derivative.
\begin{proposition}{\rm(\cite{PGLecture})}
The Poisson manifold $(G,\pi_G)$ is a Poisson Lie group if and only if $(\fg,d_e\pi_G)$ is a Lie bialgebra.
\end{proposition}
Let $\g$ be a Lie algebra and suppose that $[\ ,\ ]_{\nu^*}$ is a Lie bracket on $\fg^*$, whose dual is denoted by $\nu: \g\rightarrow \wedge^2\g$.
Let $\langle\ ,\ \rangle$ be the non-degenerate symmetric bilinear form on the direct sum vector space $\g\oplus\g^*$ given by
\[\langle x+\xi,y+\eta\rangle=(x,\eta)+(y,\xi),\ x,y\in\g,\ \xi,\eta\in \g^*,\]
where $(\ ,\ )$ denotes the pairing between $\g$ and $\g^*$.
For $x\in\g$ and $\xi\in \g^*$, denote $\mathrm{ad}_x^*\xi\in\g^*$ and $\mathrm{ad}_\xi^*x\in\g$, respectively by
\[\langle \mathrm{ad}_x^*\xi,y\rangle=\langle \xi,[y,x]\rangle,\ \langle \mathrm{ad}_\xi^*,\eta\rangle=\langle x,[\eta,\xi]_{\nu^*}\rangle,\ y\in\g,\eta\in \g^*.\]
Define the linear map on $\fg\oplus \fg^*$ by
\begin{equation}\label{eq:Manin}
[x+\xi,y+\eta]=[x,y]-\mathrm{ad}_\eta^*x+\mathrm{ad}^*_\xi y+[\xi,\eta]_{\nu^*}+\mathrm{ad}_x^*\eta-\mathrm{ad}_y^*\xi,
\end{equation}
where $x,y\in \g$ and $\xi,\eta\in \g^*$. It is remarkable that
\begin{proposition}{\rm(\cite{PGLecture})}\label{Pro:PM}
The pair $(g,\nu)$ is a Lie bialgebra if and only if $((\g\oplus \g^*,\langle,\rangle),\g,\g^*)$ is a Manin triple.
\end{proposition}
\begin{remark}
The bracket $[\ ,\ ]$  on $\g\oplus \g^*$ uniquely extends the given Lie bracket on $\g$ and $\g^*$ such that $\langle [a,b],c\rangle+\langle b,[a,c]\rangle=0,\ \forall~ a,b,c\in \g\oplus \g^*$.
\end{remark}
\begin{corollary}\label{11}
A Manin triple is one-one corresponding to a simply-connected Poisson Lie group up to Poisson isomorphism.
\end{corollary}

Let $(G,\pi_G)$ be a simply-connected Poisson Lie group, whose induced Manin triple is denoted by
$\mathcal{M}_G=((\g\oplus \g^*,\langle,\rangle),\g,\g^*)$ in Proposition \ref{Pro:PM} for $\nu=d_e\pi_G$.
\begin{definition}
The unique simply connected Poisson Lie group  $(G^{(n)},\pi_{G^{(n)}})$, whose induced Manin triple is the $n$-uble of  $\mathcal{M}_G$  and containing $G$ as a Lie subgroup is called the {\rm \textbf{$n$-uble of}} \textbf{$(G,\pi_G)$}.
\end{definition}
\begin{remark}
{\rm (1)}$(G^{(2)},\pi_{G^{(2)}})$ is called the Drinfeld double of Poisson Lie group $(G,\pi_G)$.\\
{\rm (2)}$G^{(n)}$ is locally but not globally isomorphic $G^n$ in general.
\end{remark}

The following proposition is follows directly from Theorem A and Corollary \ref{11}.
\begin{proposition}\label{isoPL}
The Poisson Lie group $(G^{(mn)},\pi_{G^{(mn)}})$ is Poisson isomorphic \\
to $((G^{(m)})^{(n)},\pi_{(G^{(m)})^{(n)}})$.
\end{proposition}

\begin{definition}
A Poisson Lie group $(G,\pi_G)$ with Lie algebra $\mathfrak{g}$ is said to be \textbf{co-boundary}
if $\pi_G=\Lambda^L-\Lambda^R$  for some $\Lambda\in \wedge^2 \mathfrak{g}$, where $\Lambda^L$ and $\Lambda^R$ are respectively the left and right invariant bivector fields associated to $\Lambda$.
\end{definition}
Let $D$ be a simply-connected Lie group with Lie algebra $\fd$. Suppose that $\mathcal{M}=((\mathfrak{d},\langle\cdot,\cdot\rangle),\mathfrak{g}_1,\mathfrak{g}_2)$ is a Manin triple.
Let $i^{mn}_n:\mathcal{M}^{mn}\rightarrow (\mathcal{M}^{m})^n$ be the unique isomorphism in Theorem \ref{tm:mn-uble}. Denote $I^{mn}_n:D^{mn}\rightarrow (D^{m})^n$ the induced Lie group isomorphism induced by $i^{mn}_n$. Let $\Lambda_{mn}$, $\Lambda_{m,n}$ be the skew-symmetric part of $r$-matrix (\ref{basis}) induced by $\mathcal{M}^{mn}$ and $(\mathcal{M}^{m})^n$ respectively. The following statement is the corollary of Proposition \ref{isoPL}.
\begin{corollary}\label{coisoPL}
The map $I^{mn}_n:(D^{mn},\Lambda^L_{mn}-\Lambda^R_{mn})\rightarrow ((D^{m})^n,\Lambda^L_{m,n}-\Lambda^R_{m,n})$ is a Poisson isomorphism.
\end{corollary}

\subsection{Poisson homogeneous spaces induced by Manin triple and Lie algebra action}
Let $(G,\pi_G)$ be a Poisson Lie group and $(P,\pi_P)$ a Poisson manifold and a homogeneous space of $G$. Recall from \cite{Drin} that if $M$ is a homogeneous space of $G$ with respect to Lie group action $G\times P\rightarrow P$ and the action map $(G\times P,\pi_G\times \pi_P)\rightarrow (P,\pi_P)$ is a Poisson map, then $(P,\pi_P)$ is called a {\it Poisson homogeneous space} of $(G,\pi_G)$.

Let $(\fd,\langle,\rangle_{\fd})$ be a quadratic Lie algebra and  $\fd=\fg_1+\fg_2$ a Lagrangian splitting. Let $M$ be a homogeneous space of $D$ with respect to Lie group action $\sigma:D\times M\rightarrow M$, whose induced Lie algebra action is also denoted by $\sigma$. Suppose that for each $m\in M$, the stabilizer subalgebra $\mathfrak{q}_{m}$ satisfies the admissible condition $[\mathfrak{q}_{m}^{\perp},\mathfrak{q}_{m}^{\perp}]\subset \mathfrak{q}_{m}$ in Remark \ref{Poisson}. Let $\pi_M:=-\sigma(\Lambda_{\g_1,\g_2})$ and $\pi_D=\Lambda_{\g_1,\g_2}^L-\Lambda_{\g_1,\g_2}^R$, where $\Lambda_{\g_1,\g_2}$ is the skew-symmetric part of (\ref{basis}) and $\Lambda_{\g_1,\g_2}^L$ or $\Lambda_{\g_1,\g_2}^R$ is respectively the left or right invariant bivector fields associated to $\Lambda_{\g_1,\g_2}$. Although the condition on $\mathfrak{q}_{m}$ changes, the statement and proof of following proposition in \cite{T-leaves} still follows.
\begin{proposition}
The Poisson manifold $(M,\pi_M)$ is a Poisson homogeneous space of $(D,\pi_D)$.
\end{proposition}

Let $Q$ be a closed subgroup of $D$ with Lie algebra $\fq$, which satisfied the admissible condition $[\fq^\perp,\fq^\perp]\subset \fq$.
The following statement is the special case when $M=D/Q$ .
\begin{corollary}\label{CoHo}
The Poisson manifold $(D/Q,\pi_{D/Q})$ is a Poisson homogeneous space of $(D,\pi_D)$.
\end{corollary}
\subsection{Examples of homogeneous spaces associate to complex semisimple Lie groups}\label{subsec:example}
Let $G$ be a simply-connected complex semisimple Lie group.
Fix a pair $(B,B_-)$ of opposite Borel subgroups and let $T=B\cap B_-$ be the maximal torus defined by $B$ and $B_-$.
Let $N$ and $N_-$ be the unipotent radical of $B$ and $B_-$ respectively.
Let $N_G(T)$ be the normalizer subgroup of $T$ in $G$ and let $W=N_G(T)/T$ be the Weyl group of $G$.
\begin{example}\label{example}
Let $Q$ be a closed subgroup of $G$ containing $N$.
Suppose that $M=G/Q\cong (G\times T)/(Q\times T)$ and define $\sigma:$left translation of $G\times T$ on $M$.
Consider also the Manin triple
\begin{multline*}
\mathfrak{d}=\mathfrak{g}\oplus \mathfrak{h};\ \ \  \langle(x_1,y_1),(x_2,y_2)\rangle=\langle x_1,x_2\rangle_\mathfrak{g}-\langle y_1,y_2\rangle_\mathfrak{g},\\
 \mathfrak{g}_1=\{(x_++x_1,x_1)|x_1\in \mathfrak{h},x_{+}\in \mathfrak{n}\},
\mathfrak{g}_2=\{(x_-+x_2,-x_2)|x_2\in \mathfrak{h},x_{-}\in \mathfrak{n}_{-}\}.
\end{multline*}
Notice that the Lie algebra $\q$ of $Q$ such that
\[ [(\q\oplus \h)^\perp,(\q\oplus \h)^\perp] \subset [(\n\oplus \h)^\perp,(\n\oplus \h)^\perp]=\n\oplus\h\subset \q\oplus \h.
 \]
It follows from Corollary \ref{CoHo} that $G/Q\cong (G\times T)/(Q\times T)$ can be equipped with a Poisson structure.
Moreover, each Lagrangian splitting of the quadratic algebra $(\mathfrak{d},\langle,\rangle)$  induces a Poisson structure $\sigma(\Lambda_{\g_1,\g_2})$ on $G/Q$.
\end{example}
\begin{remark}
When $Q=N$,  $\sigma(r_{\g_1,\g_2})$  is not a bivector field on $G/N$ but $\sigma(\Lambda_{\g_1,\g_2})$ is a Poisson structure on $G/N$.
\end{remark}

\begin{example}\label{piM}
Consider the multi-flag variety $G^n/B^n\cong (G\times T)^n/(B\times T)^n$ and let $\sigma$ be left translation of $(G\times T)^n$ on $(G\times T)^n/(B\times T)^n$. Let $\mathcal{M}$ be the Manin triple defined by
\begin{multline*}
\mathfrak{d}=\mathfrak{g}\oplus \mathfrak{h};\ \ \  \langle(x_1,y_1),(x_2,y_2)\rangle=\langle x_1,x_2\rangle_\mathfrak{g}-\langle y_1,y_2\rangle_\mathfrak{g},\\
\mathfrak{g}_1=\{(x_++x_1,x_1)|x_1\in \mathfrak{h},x_{+}\in \mathfrak{n}\}, \mathfrak{g}_2=\{(x_-+x_2,-x_2)|x_2\in \mathfrak{h},x_{-}\in \mathfrak{n}_{-}\}.
\end{multline*}
Let $\mathcal{M}^n$ be  the $n$-uble of $\mathcal{M}$. Then, $\mathcal{M}^n$ and $\sigma$ induce a Poisson structure on $G^n/B^n\cong (G\times T)^n/(B\times T)^n$, denoted by $\pi_{\mathcal M^{n}}$.
\end{example}
\begin{example}
Consider the multi-double flag variety $(G\times G)^n/(B\times B_-)^n$ and let $\sigma$ be left translation of $(G\times G)^n$ on $(G\times G)^n/(B\times B_-)^n$.
Let $\mathcal{M}$ be the Manin triple defined by
\begin{multline*}
\mathfrak{d}=\mathfrak{g}\oplus \mathfrak{g};\ \ \  \langle(x_1,y_1),(x_2,y_2)\rangle=\langle x_1,x_2\rangle_\mathfrak{g}-\langle y_1,y_2\rangle_\mathfrak{g},\\
\mathfrak{g}_1=\{(x,x)|x\in \mathfrak{g}\},
\mathfrak{g}_2=\{(x_++x_0,x_{-}-x_0)|x_0\in \mathfrak{h},x_{+}\in \mathfrak{n},x_{-}\in \mathfrak{n}_{-}\}.
 \end{multline*}
 Let $\mathcal{M}^n$ be  the $n$-uble of $\mathcal{M}$. Then, $\mathcal{M}^n$ and $\sigma$ induce a Poisson structure on $(G\times G)^n/(B\times B_-)^n$, denoted by $\Pi_{\mathcal M^{n}}$.
\end{example}
\subsection{Proof of Theorem B(2)}
Let $\mathcal{M}=((\fd,\langle,\rangle_{\fd}),\fg_1,\fg_2)$ be a Manin triple. Let $i^{mn}_n:\mathcal{M}^{mn}\rightarrow (\mathcal{M}^{m})^n$ be the isomorphism of Manin triple defined  by (\ref{tm:mn-uble}). Let $D$ be the simply-connected Lie group with Lie algebra $\fd$. Denote $I^{mn}_n:D^{mn}\rightarrow (D^{m})^n$ the induced Lie group isomorphism induced by $i^{mn}_n$. Let $Q$ be a closed subgroup of $D^{mn}$ with Lie algebra $\fq$, which satisfies the admissible condition $[\fq^\perp,\fq^\perp]\subset \fq$. Let $(D^{mn}/Q,\Pi^{mn})$  be Poisson homogeneous space  induced by $\mathcal{M}^{mn}$ defined in Corollary \ref{CoHo} and $((D^m)^n/I(Q),\Pi^{m,n})$ the Poisson homogeneous space induced by $(M^{m})^n$.
\begin{theorem}\label{PoissonIso}
The map $I^{mn}_{n,Q}: (D^{mn}/Q,\Pi^{mn}) \rightarrow ((D^m)^n/I^{mn}_n(Q),\Pi^{m,n}),\ \ [g]\mapsto [I^{mn}_n(g)]$ is an isomorphism of Poisson homogeneous spaces.
\end{theorem}
\begin{proof}
It follows from Corollary \ref{coisoPL} that the Poisson Lie group  $(D^{mn},\Lambda^L_{mn}-\Lambda^R_{mn})$ is Poisson isomorphic to $((D^{m})^n,\Lambda^L_{m,n}-\Lambda^R_{m,n})$  through $I^{mn}_{n}$. Moreover,
since the diagram
\begin{center}
\begin{tikzpicture}
  \matrix (m) [matrix of math nodes,row sep=3em,column sep=4em,minimum width=2em]
  {
     \fd^{mn} & V^1(D^{mn}/Q) \\
     (\fd^{m})^n & V^1((D^{m})^n/I^{mn}_n(Q) ) \\};
  \path[-stealth]
    (m-1-1) edge node [right] {$i^{mn}_n$} (m-2-1)
            edge  node [above] {$\sigma$} (m-1-2)
    (m-2-1) edge node [below] {$\sigma$}
            (m-2-2)
    (m-1-2) edge node [right] {$(I^{mn}_{n,Q})_*$} (m-2-2)
            (m-2-1);
\end{tikzpicture}
\end{center}
is commutative, the Poisson structure $(I^{mn}_{n,Q})_*(\Pi^{mn})$ coincides with $\Pi^{m,n}$.
\end{proof}
\section{Identification of $(DF_n,\Pi_n)$ as $(F_{2n},\pi_{2n})$}
\subsection{Notation and preliminary}\label{Notation}
Throughout introduction and this part, if $A$ is a group, $n \geq 1$
is an integer, and $Q_1, \ldots,  Q_n$ are subgroups of $A$, let
\begin{equation}\label{ta}
A \times_{Q_1}  \cdots \times_{Q_{n-1}} A/Q_n
\end{equation}
denote the quotient of $A^n$ by the  action of the subgroup
$Q_1  \times \cdots \times Q_{n} \subset G^n$,  where the $n$-fold product group $G^n$ acts  on
itself (freely) from the right  by
\begin{equation}\lb{eq-Gn-Gn-1}
(g_1^\prime, \ldots, g_n^\prime) \cdot_r (g_1, \ldots, g_{n}) = (g_1^\prime g_1, \, g_1^{-1}g_2^\prime g_2, \; \ldots, \; g_{n-1}^{-1}g_n^\prime g_n),
\hs g_j, g_j^\prime \in G, \; 1 \leq j \leq n.
\end{equation}
If the quotient space $A \times_{Q_1}  \cdots \times_{Q_{n-1}} A/Q_n$ is denoted by $Z$, the image of
$(g_1, \ldots, g_n) \in A^n$ in $Z$ will be denoted by $[g_1, g_2, \ldots, g_n]_{\sZ}$, and the projection from $A^n$ to $Z$
will be denoted by
\[
\varpi_{\sZ}: \; A^n \lrw Z, \;\;\; (g_1, g_2, \ldots, g_n) \longmapsto [g_1, g_2, \ldots, g_n]_{\sZ}.
\]
If $S_1, S_2, \ldots, S_n$ are subsets of $G$ such that $S_1$ is right $Q_1$-invariant and
$S_j$ is left $Q_{j-1}$ and right $Q_j$
invariant for $2 \leq j \leq n$, we set
\[
S_1 \times_{Q_1} \cdots \times_{Q_{n-1}} S_n/Q_n = \varpi_{\sZ}(S_1 \times \cdots \times S_n) \subset Z.
\]
If $Q_n = \{e\}$ is the trivial subgroup, we denote $A \times_{Q_1}  \cdots \times_{Q_{n-1}} A/Q_n$ by
$G \times_{Q_1}  \cdots \times_{Q_{n-1}} G$.

When $(A, \pi_{\sA})$ is a Poisson Lie group and when the $Q_j$'s are
closed Poisson Lie subgroups of $(A, \pi_{\sA})$,
the $n$-fold product Poisson structure $\pi_{\sA}^n$ on $A^n$ projects to a well-defined
Poisson structure, referred to as the {\it quotient Poisson
structure}, on $A \times_{Q_1}  \cdots \times_{Q_{n-1}} A/Q_n$ (see \cite[$\S$7]{mixed}).
\subsection{Construction of $(DF_n,\Pi_n)$ and $(F_{2n},\pi_{2n})$ }\label{con}
Let $G$ be a simply-connected complex semisimple Lie group.
Fix a pair $(B,B_-)$ of opposite Borel subgroups and let $T=B\cap B_-$ be the maximal torus defined by $B$ and $B_-$.
Let $N$ and $N_-$ be the unipotent radical of $B$ and $B_-$ respectively.
Let $N_G(T)$ be the normalizer subgroup of $T$ in $G$ and let $W=N_G(T)/T$ be the Weyl group of $G$.

Let $\fg$ be the Lie algebra of $G$ and let $\langle,\rangle_\fg$ be a fixed multiple of the killing form of $\fg$.
Let $\Delta$ be the root system defined by $T$. The sets of positive roots and simple roots determined by $B$  are denoted by $\Delta_+$ and $\Gamma$ respectively.
Let $\fg=\fh+\sum_{\alpha\in \Delta}\fg_\alpha$ be the root space decomposition.
For $\alpha\in\Delta_+$, let $H_\alpha\in\fh$ be such that $\alpha(x)=\langle x, H_\alpha \rangle_\fg$.
We also fix root vectors $E_\alpha\in \fg_\alpha$ and $E_{-\alpha}\in \fg_{-\alpha}$ such that $[E_{-\alpha},E_\alpha]=H_\alpha$.
Define
\[\Lambda_{\rm st}=\frac{1}{2}\sum_{\alpha\in\Delta_+}E_{-\alpha}\wedge E_{\alpha},\]
and
\[\pi_{\rm st}=\Lambda^L-\Lambda^R.\]
Then $(G,\pist)$ is a Poisson Lie group. The Drinfeld double of $(G,\pist)$ is denoted by $(G\times G,\Pist)$.

When taking $(A,\pi_A)=(G,\pist)$ and $Q_1=Q_2=...=Q_{n}=B$, the quotient of $G ^{n}$ by the right action of $ B ^{n}$ is denoted by
\begin{equation}\label{fn}
F_{n}:=G\times_B G\times_B \cdots \times_BG/B
\end{equation}
and the quotient Poisson structure is denoted by $\pi_{n}$.

When taking $(A,\pi_A)=(G\times G,\Pist)$ and $Q_1=Q_2=...=Q_{n}=B\times B_-$.
The quotient of $(G\times G)^{n}$ by the right action of $(B\times B_-)^{n}$ is denoted by
\begin{equation}\label{dfn}
DF_{n}:=(G\times G)\times_{B\times B_-} (G\times G)\times_{B\times B_-} \cdots \times_{B\times B_-} (G\times G)/B\times B_-
\end{equation}
and the quotient Poisson structure is denoted by $\Pi_{n}$.

\subsection{Proof of Theorem C}
\begin{theorem}\label{theo:iso}
Mult-flag variety $(F_{2n},\pi_{2n})$ is Poisson isomorphic to $(DF_n,\Pi_n)$ through the map
\begin{equation}\label{Poiso}
\Psi_n:[g_1,g_2,\ldots,g_{2n}]\mapsto [(g_1,g_1g_2\ldots g_{2n}w_0),\\
 (g_2,w_0 g^{-1}_{2n}w_0),\ldots,(g_n,w_0 g^{-1}_{n+2}w_0) ],
\end{equation}
where $w_0$ is the longest element in the Weyl group $W$.
\end{theorem}
\begin{proof}
Decompose  $\Psi_n=I_4\circ I_3\circ I_2\circ I_1$ as follow:
\begin{align}\label{eq:comp-1}
(F_{2n}, \; \pi_{2n}) \; \; \stackrel{I_1 }{\!\!\!-\!\!\!\longrightarrow}\;
(G^{2n}/B^{2n}, \; -\pi_{\mathcal M^{2n}})\
\stackrel{I_2}{\!\!\!-\!\!\!-\!\!\!\longrightarrow}\;
 (G^{2n}/(B^n\times B_-^n),-\pi^\pm_{\mathcal M^{2n}} ) \\
 \stackrel{I_3 }{\!\!\!-\!\!\!\longrightarrow}\; ((G\times G)^n/(B\times B_-)^n, \; -\Pi_{\mathcal M^{n}})\; \stackrel{I_4 }{\!\!\!-\!\!\!\longrightarrow}\; (DF_n, \; \Pi_{n})
\end{align}
in which the map is defined by
\begin{align*}
[g_1,g_2,\ldots,g_{2n}] \; \; \stackrel{I_1 }{\longmapsto}\;
 (g_1,g_1g_2,\ldots,g_1g_2\ldots g_{2n})\;\\
\stackrel{I_2}{\longmapsto}\;
 (g_1,g_1g_2,\ldots,g_1g_2\ldots g_n,g_1\ldots g_{n+1}w_0,g_1\ldots g_{2n}w_0) \\
 \stackrel{I_3 }{\longmapsto}\; ((g_1,g_1\ldots g_{2n}w_0),(g_1g_2,g_1\ldots g_{2n-1}w_0),\ldots (g_1\ldots g_{n},g_1\ldots g_{n+1}w_0))\;\\ \stackrel{I_4 }{\longmapsto}\; [(g_1,g_1\ldots g_{2n}w_0),(g_2,w_0g_{2n}^{-1}w_0),\ldots,(g_n,w_0g_{n+2}^{-1}w_0)]
\end{align*}
where $\pi_{\mathcal M^{2n}}$ is the Poisson structure defined in $\ref{piM}$ and $\pi^\pm_{\mathcal M^{2n}}$ is the Poisson structure on $G^{2n}/(B^n\times B_-^n)\cong (G\times T)^{2n}/(B\times T)^n\times (B_-\times T)^n$ induced by 2n-uble $\mathcal M^{2n}$ and left translation of $(G\times T)^{2n}$ on $(G\times T)^{2n}/(B\times T)^n\times (B_-\times T)^n$. Since the map
\[G/B\rightarrow G/B_-,\ \ g.B\rightarrow gw_0.B_-\]
is a $G$-equivariant isomorphism and both $\pi_{\mathcal M^{2n}}$ and $\pi^\pm_{\mathcal M^{2n}}$ are induced by the same Manin triple, $I_3$ is Poisson isomorphism. In addition,
it follows from \cite[Theorem 1.2]{mixed} that both $I_1$ and $I_4$ are Poisson. Moreover, it follows from Theorem B(2) that $I_3$ is Poisson. Therefore, $\Psi_n$ is a Poisson isomorphism.
\end{proof}
\subsection{$T$-leaves identification of $(DF_n,\Pi_n)$}\label{se:Tleaf}
Let $\mathbb{T}$ be a complex torus. A $\mathbb{T}$-Poisson manifold, defined in \cite[Section 1.1]{T-leaves},  is a complex Poisson manifold $(X,\pi_X)$ with a $\mathbb{T}$-action preserving the Poisson structure
$\pi_X$.
\begin{definition}{\rm(\cite[Section 2.2]{T-leaves})}
Let $(X,\pi_X)$ be a $\mathbb{T}$-Poisson manifold. Suppose that $\Sigma$ is a symplectic leaf in $(X,\pi_X)$. The set
\[
\mathbb{T}\Sigma=\bigsqcup_{t\in \mathbb{T}}t\Sigma
\]
is called \emph{a single $\mathbb{T}$-leaf in $(X,\pi_X)$}.
\end{definition}

For $\bfu=(u_1,\ldots,u_{n})\in W^n$ and $w\in W$, define
\[
R^\bfu_w=(B\bfu B/B)\cap \mu_{F_n}^{-1}(B_-wB/B)\subset F_n
\]
For $\bfu=(u_1,\ldots,u_{n}),\bfv=(v_1,\ldots,v_{n})\in W^n$ and $w\in W$, let
$G(w)$ be the $G_\Delta$-orbit through the point $(w,e).(B\times B_-)\in (G\times G)/(B\times B_-),$ where $G_\Delta=\{(g,g):g\in G\}\subset G\times G$, and define
\[
R^{\bfu,\bfv}_w=((B\times B_-)(\bfu,\bfv)(B\times B_-)/(B\times B_-))\cap \mu_{DF_n}^{-1}(G(w))\subset DF_n.
\]
\begin{theorem}{\rm(\cite{T-leaves})}
The decomposition of $F_n$ into $T$-leaves of the Poisson structure $\pi_n$ is
\[
F_n=\bigsqcup_{\bfu\in W^n,w\in W}R^\bfu_w.
\]
\end{theorem}
The following proposition gives the single $T$-leaf corresponding between $(DF_n,\Pi_{n})$ and $(F_{2n},\pi_{2n})$.

\begin{theorem}{\rm(\cite{T-leaves})}
The decomposition of $DF_n$ into $T$-leaves of the Poisson structure $\Pi_n$ is
\[
DF_n=\bigsqcup_{\bfu,\bfv\in W^n,w\in W}R^\bfu_w.
\]
\end{theorem}
The following proposition gives the single $T$-leaf corresponding between $(DF_n,\Pi_{n})$ and $(F_{2n},\pi_{2n})$.

\begin{proposition}
The single $T_\Delta$-leaf
\[(R^{(u_1,..,u_n),(v_1,...,v_n)}_w,\Pi_n)\]
of $(DF_n,\Pi_{n})$ is Poisson isomorphic to the single $T$-leaf
\[(R^{(u_1,ww_0,u_2,w_0v_{n}^{-1}w_0,...,u_n,w_0v_{2}^{-1}w_0)}_{v_1w_0},\pi_{2n})\]
of $(F_{2n},\pi_{2n})$ via $\Psi_n$.
\end{proposition}
\begin{proof}
It follows from Theorem \ref{theo:iso} that: the pre-image $\Psi^{-1}_n((R^{(u_1,..,u_n),(v_1,...,v_n)}_w) $ in $DF_n$ is Poisson isomorphic to
$R^{(u_1,ww_0,u_2,w_0v_{n}^{-1}w_0,...,u_n,w_0v_{2}^{-1}w_0)}_{v_1w_0}$ in $F_{2n}$.
\end{proof}

\end{document}